%% file: main.tex
\documentclass{eptcs}
\providecommand{\event}{LICS '23}

\input{preamble/packages}
\input{preamble/symbols}

\usepackage{multicol}

\title{Completeness for arbitrary finite dimensions of ZXW-calculus, a unifying calculus}
\def\titlerunning{Completeness for arbitrary finite dimensions of ZXW-calculus}
\date{}

\author{
  Boldizsár Poór$\null^{1}$ \\[5pt]
  Lia Yeh$\null^{1,2}$ \and
  Quanlong Wang$\null^{1}$ \\[5pt]
  Richie Yeung$\null^{1,2}$ \and
  Razin A. Shaikh$\null^{1,2}$ \\[5pt]
  Bob Coecke$\null^{1}$ \and
  \institute{$\null^{1}$Quantinuum, 17 Beaumont Street, Oxford, OX1 2NA, United Kingdom}
  \institute{$\null^{2}$University of Oxford, Oxford, United Kingdom}
}
\def\authorrunning{B. Poór, Q. Wang, R. A. Shaikh, L. Yeh, R. Yeung, and B. Coecke}

\begin{document}

\maketitle

\begin{abstract}
  \input{abstract.tex}
\end{abstract}

\section{Introduction}\label{sec:introduction}
\input{introduction.tex}

\section{ZXW-Calculus}\label{sec:zxw_calculus}
\input{zxw-calculus.tex}

\section{Completeness}\label{sec:completeness}
\input{proofidea.tex}
\input{completenessnoproof.tex}

\section{Conclusion and further work}\label{sec:future}
\input{conclusion.tex}

\section*{Acknowledgements}
We would like to thank Giovanni De Felice, Mark Koch, furthermore the anonymous reviewers for their detailed feedback and their numerous suggestions for improvement.
RS is supported by the Clarendon Fund Scholarship.
LY is supported by the Basil Reeve Graduate Scholarship at Oriel College with the Clarendon Fund.
RY thanks Simon Harrison for his generous support for the Wolfson Harrison UK Research Council Quantum Foundation Scholarship that he enjoys.

\bibliographystyle{eptcs}
\bibliography{preamble/references-bibtex}

\onecolumn
\appendix

\section{Appendix}

First we prove a series of lemmas using the rule set given in \cref{subsec:rule-set}.
Then, based on these lemmas, we show the proof of the main lemmas that are necessary to prove the completeness of the ZXW-calculus for any finite dimension.

\allowdisplaybreaks
\setlength{\jot}{20pt}
\input{lemmas.tex}
\input{appendix.tex}

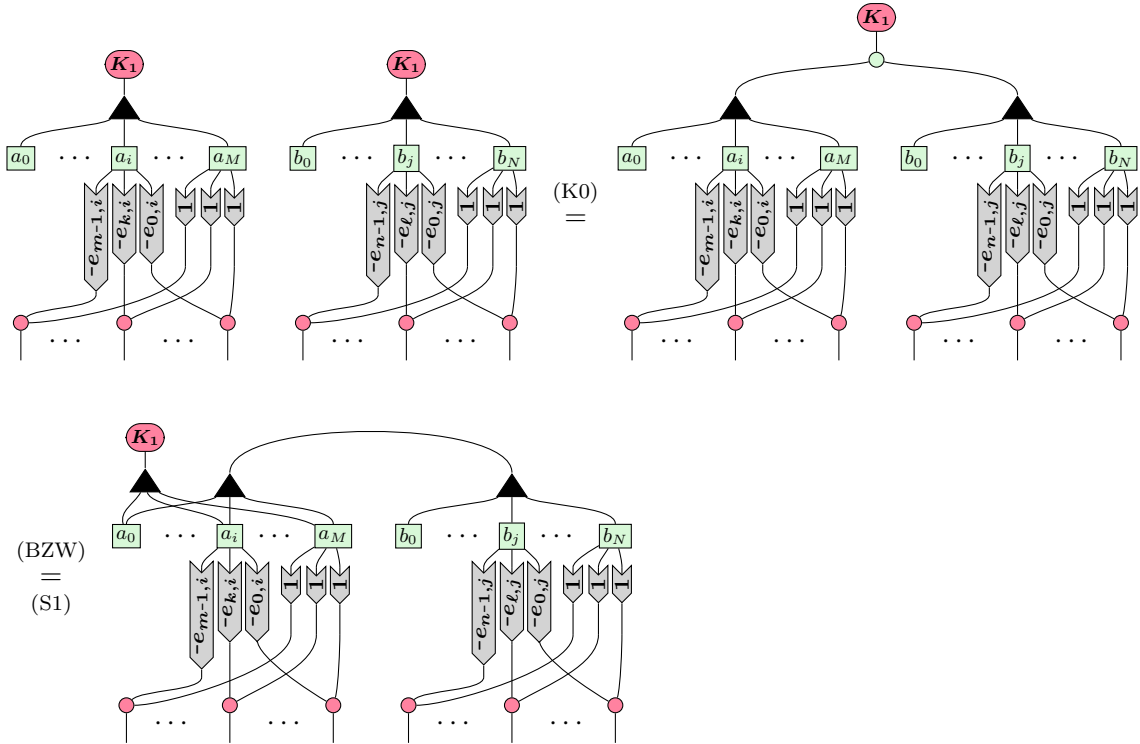
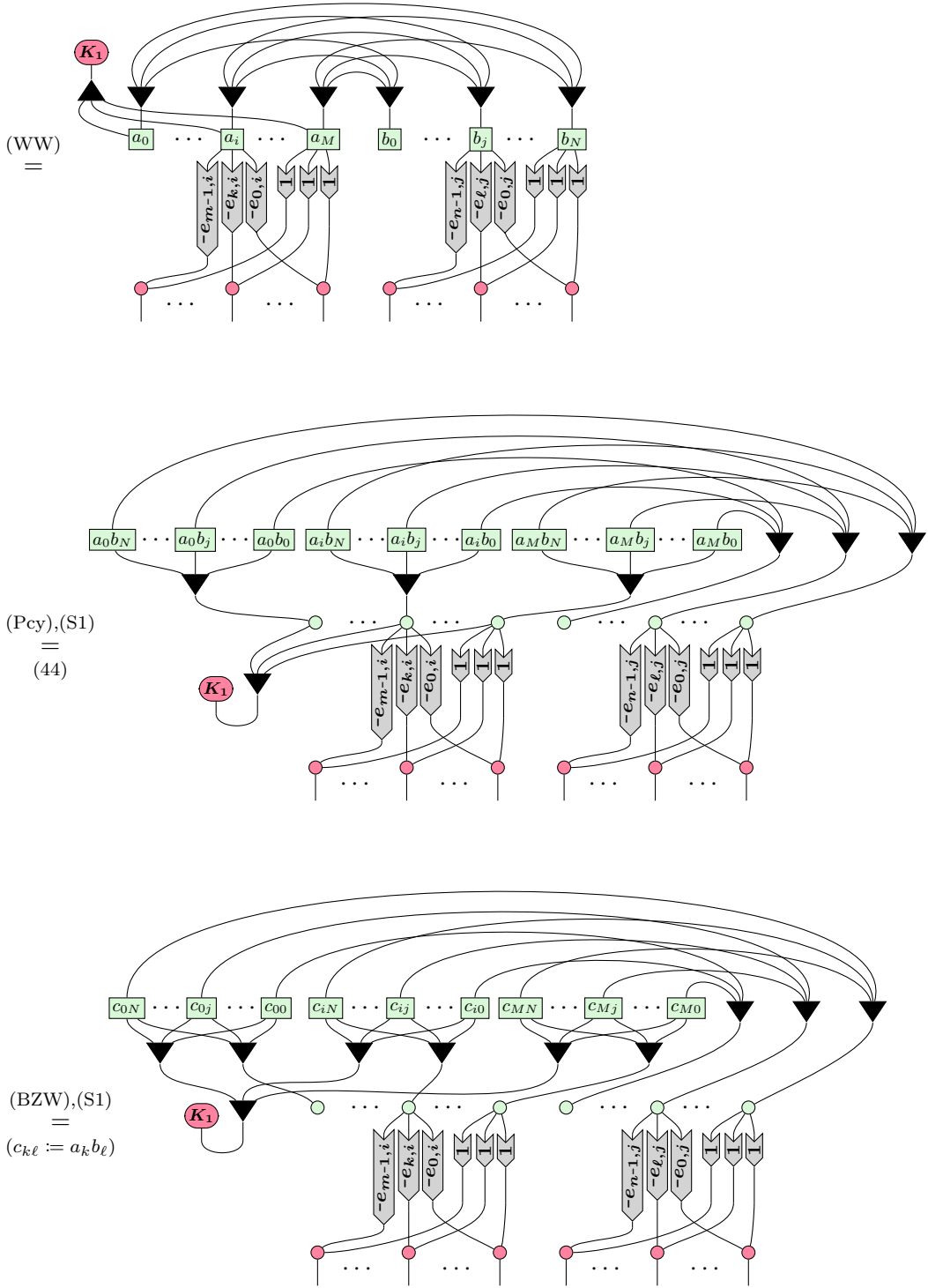
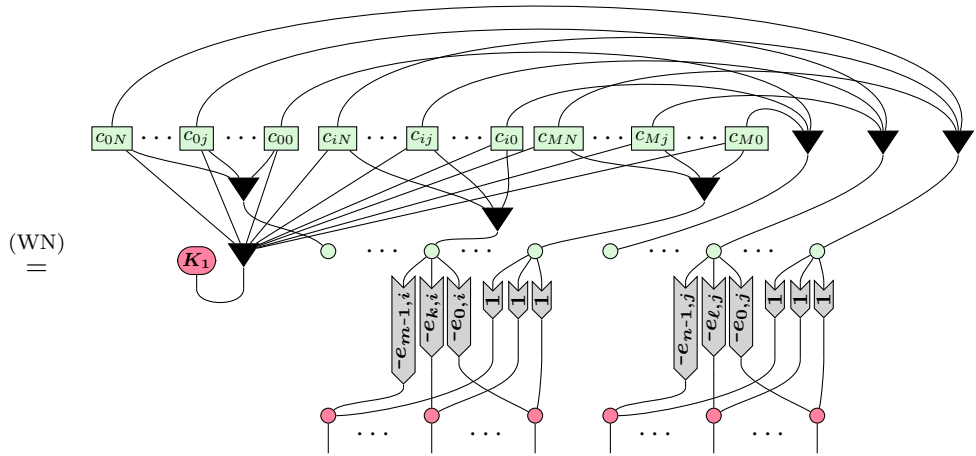
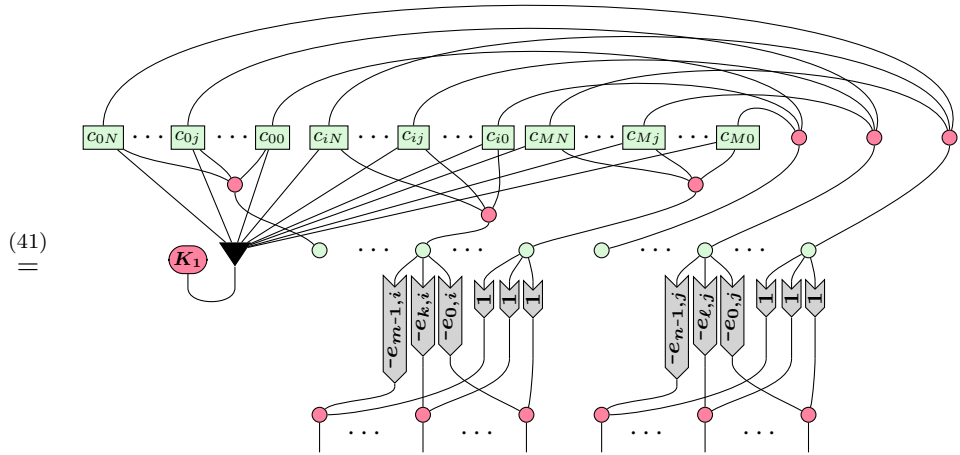
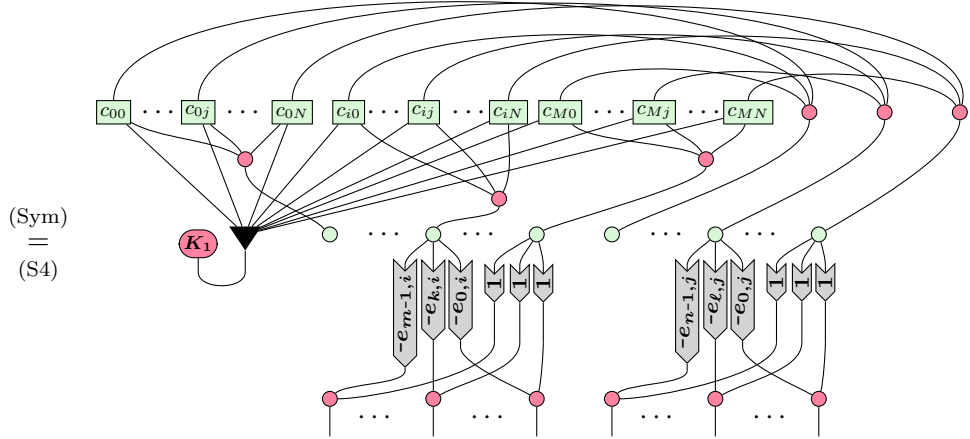
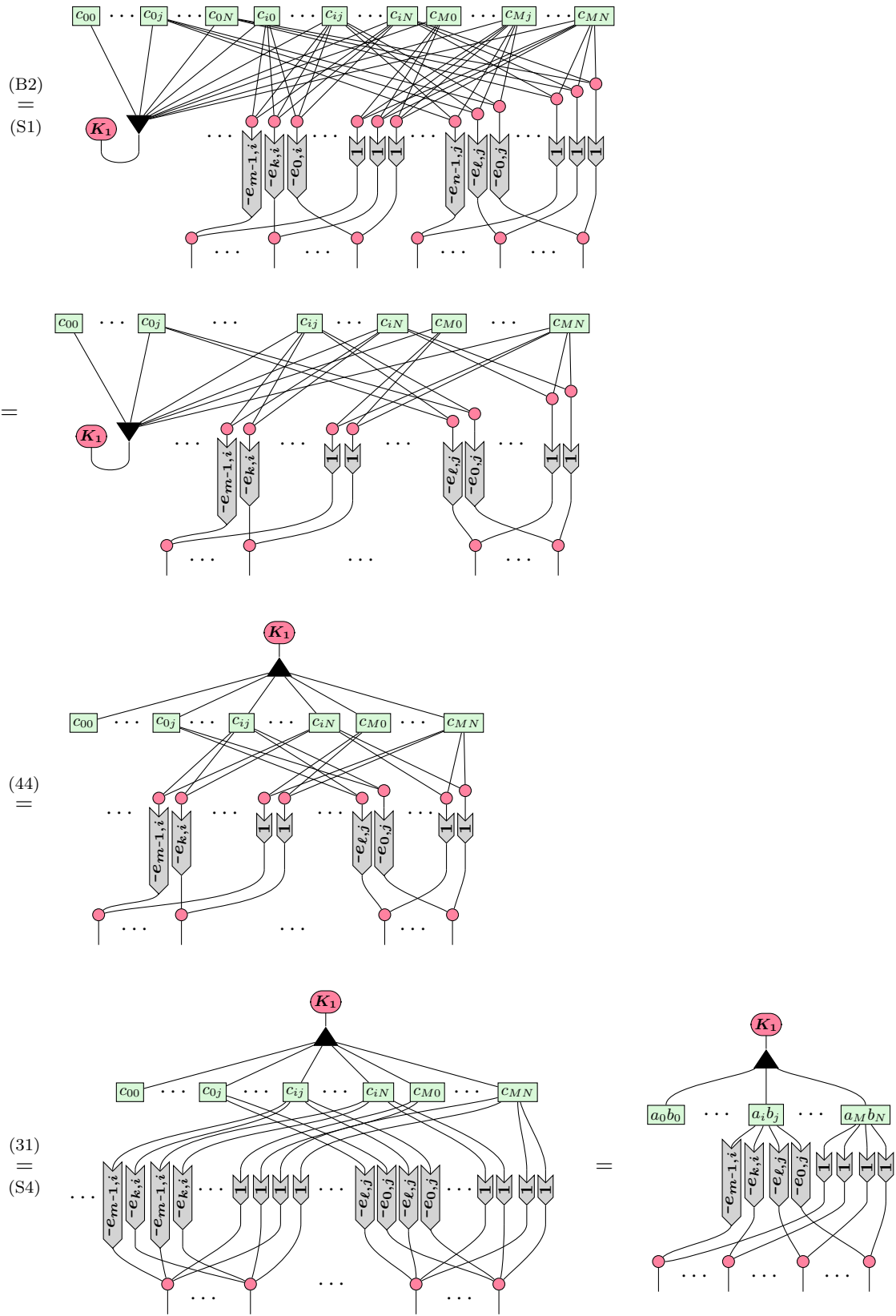
\end{document}

%% file: preamble/packages.tex
\usepackage[UKenglish]{babel}
\usepackage{fontenc}
\usepackage{colorprofiles}
\usepackage[a-3b,mathxmp]{pdfx}[2018/12/22]
\usepackage{bookmark}   
\usepackage[doipre={doi:}]{uri}   
\hypersetup{hidelinks}

\usepackage{csquotes}

\usepackage{etoolbox}
\apptocmd{\sloppy}{\hbadness 10000\relax}{}{} 

\usepackage{mathtools}
\usepackage{physics}
\usepackage{amsfonts}
\usepackage{amssymb}
\usepackage{mathdots}   
\usepackage{mathrsfs}   
\usepackage{bigints}    
\usepackage{xfrac}      
\usepackage{stmaryrd}   
\SetSymbolFont{stmry}{bold}{U}{stmry}{m}{n}   

\usepackage{amsthm}
\usepackage{thmtools}
\usepackage{thm-restate}

\newtheorem{lemma}{Lemma}
\newtheorem{remark}{Remark}

\newtheorem{definition}{Definition}

\usepackage{placeins}

\usepackage{bm}   

\usepackage{blkarray}   

\usepackage{graphicx}
\graphicspath{figures}

\usepackage{tikzfig}
\usepackage{pgfplots} 
\pgfplotsset{compat=newest, compat/show suggested version=false} 

\input{preamble/zxw.tikzstyles}
\definecolor{zx_grey}{RGB}{211,211,211}
\definecolor{zx_pink}{RGB}{255, 130, 160}
\definecolor{zx_green}{RGB}{216,248,216}

\usepackage{xtab}
\usepackage{tabu}

\usepackage[capitalise, noabbrev]{cleveref}

%% file: preamble/zxw.tikzstyles
\tikzstyle{gbox}=[box, draw=black, shape=rectangle, fill={zx_green}]
\tikzstyle{hbox}=[box, draw=black, shape=rectangle, fill=yellow, minimum size=.55em]
\tikzstyle{box}=[draw=black, shape=rectangle, fill=white, minimum size=.95em, inner sep=0.15em, scale=0.85, font={\footnotesize}]
\tikzstyle{gn}=[draw=black, shape=circle, fill={zx_green}, draw=black, inner sep=0.7mm, minimum width=0pt, minimum height=0pt]
\tikzstyle{rn}=[gn, fill={zx_pink}, draw=black, tikzit fill={rgb,255: red,232; green,165; blue,165}]
\tikzstyle{gn_phase}=[shape=rectangle, fill={zx_green}, draw=black, minimum size=1.2em, rounded corners=0.5em, inner sep=0.2em, outer sep=-0.2em, scale=0.8, font={\footnotesize\boldmath}, tikzit shape=circle, tikzit fill={rgb,255: red,181; green,215; blue,181}]
\tikzstyle{rn_phase}=[{gn_phase}, fill={zx_pink}, draw=black, tikzit fill={rgb,255: red,215; green,96; blue,96}]
\tikzstyle{rtriang}=[shape=isosceles triangle, fill=yellow, draw=black, isosceles triangle stretches=true, inner sep=0.8pt, minimum width=0.25cm, minimum height=2mm]
\tikzstyle{ltriang}=[rtriang, shape=isosceles triangle, fill=yellow, draw=black, shape border rotate=180]
\tikzstyle{utriang}=[rtriang, shape=isosceles triangle, fill=yellow, draw=black, shape border rotate=90]
\tikzstyle{dtriang}=[rtriang, shape=isosceles triangle, fill=yellow, draw=black, shape border rotate=-90]
\tikzstyle{lmat}=[shape=signal, signal to=west, signal from=east, fill={zx_grey}, draw=black, minimum height=6pt, inner sep=1pt, font={\scriptsize  \boldmath}, tikzit fill=gray, tikzit category=GLA]
\tikzstyle{rmat}=[shape=signal, signal to=east, signal from=west, fill={zx_grey}, draw=black, minimum height=6pt, inner sep=1pt, font={\scriptsize  \boldmath}, tikzit fill=gray, tikzit category=GLA]
\tikzstyle{dmat}=[shape=signal, signal to=west, signal from=east, fill={zx_grey}, draw=black, minimum height=6pt, inner sep=1pt, font={\scriptsize  \boldmath}, tikzit fill=gray, tikzit category=GLA, rotate=90]
\tikzstyle{umat}=[shape=signal, signal to=east, signal from=west, fill={zx_grey}, draw=black, minimum height=6pt, inner sep=1pt, font={\scriptsize  \boldmath}, tikzit fill=gray, tikzit category=GLA, rotate=90]
\tikzstyle{uw}=[shape=isosceles triangle, isosceles triangle stretches=true, fill=black, draw=black, minimum width=0.4cm, minimum height=3mm, inner sep=1pt, shape border rotate=90]
\tikzstyle{dw}=[shape=isosceles triangle, isosceles triangle stretches=true, fill=black, draw=black, minimum width=0.4cm, minimum height=3mm, inner sep=1pt, shape border rotate=-90]
\tikzstyle{d_split}=[shape=trapezium, fill=white, draw=black, minimum width=11pt, inner sep=1pt]
\tikzstyle{d_merge}=[shape=trapezium, fill=white, draw=black, minimum width=11pt, inner sep=1pt, rotate=180]
\tikzstyle{braceedge}=[decorate, decoration={brace, amplitude=2mm, raise=-1mm}]
\tikzstyle{wire label}=[font={\scriptsize}, auto]


%% file: preamble/symbols.tex
\newcommand{\Z}{\mathbb{Z}}

\ifdef{\C}
{\renewcommand{\C}{\mathbb{C}}}
{\newcommand{\C}{\mathbb{C}}}
\newcommand{\minu}{\texttt{-}}
\newcommand{\plus}{\texttt{+}}

\newcommand{\interp}[1]{\left\llbracket#1\right\rrbracket}

\newcommand{\bR}{\begin{color}{red}}
\newcommand{\bB}{\begin{color}{blue}}
\newcommand{\bM}{\begin{color}{magenta}}
\newcommand{\bC}{\begin{color}{cyan}}
\newcommand{\bW}{\begin{color}{white}}
\newcommand{\bBl}{\begin{color}{black}}
\newcommand{\bG}{\begin{color}{green}}
\newcommand{\bY}{\begin{color}{yellow}}
\newcommand{\e}{\end{color}}

\newcommand{\tikzrefsize}[1]{\scriptsize{#1}}
\newcommand{\axref}[1]{\tikzrefsize{\eqref{rule:#1}}}
\newcommand{\lemref}[1]{\tikzrefsize{\eqref{#1}}}

\newcommand{\Mod}[1]{\ (\mathrm{mod}\ #1)}

%% file: abstract.tex
The ZX-calculus is a universal graphical language for qubit quantum computation, meaning that every linear map between qubits can be expressed in the ZX-calculus.
Furthermore, it is a complete graphical rewrite system: any equation involving linear maps that is derivable in the Hilbert space formalism for quantum theory can also be derived in the calculus by rewriting.
It has widespread usage within quantum industry and academia for a variety of tasks such as quantum circuit optimisation, error-correction, and education.

The ZW-calculus is an alternative universal graphical language that is also complete for qubit quantum computing.
In fact, its completeness was used to prove that the ZX-calculus is universally complete.
This calculus has advanced how quantum circuits are compiled into photonic hardware architectures in the industry.

Recently, by combining these two calculi, a new calculus has emerged for qubit quantum computation, the \emph{ZXW-calculus}.
Using this calculus, graphical\mbox{-}differentiation, \mbox{-}integration, and \mbox{-}exponentiation were made possible, thus enabling the development of novel techniques in the domains of quantum machine learning and quantum chemistry.

Here, we generalise the ZXW-calculus to arbitrary finite dimensions, that is, to \emph{qudits}.
Moreover, we prove that this graphical rewrite system is complete for any finite dimension.
This is the first completeness result for any universal graphical language beyond qubits.

%% file: introduction.tex
Graphical languages~\cite{selingerSurveyGraphicalLanguages2011} have become an influential tool in computer science~\cite{bonchiStringDiagramRewrite2022a} and in quantum computing~\cite{abramskyCategoricalQuantumMechanics2008, heunenCategoriesQuantumTheory2019}.
The ZX-calculus~\cite{coeckeInteractingQuantumObservables2008, coeckeInteractingQuantumObservables2011} is one such graphical language for reasoning about quantum computation.
It has become an immensely useful tool in the industry and is used in major quantum computing companies such as Quantinuum, IBM, Google, PsiQuantum, Quandela, and many more.
Since its introduction, the ZX-calculus has been successful in many areas such as quantum circuit optimisation~\cite{debeaudrapFastEffectiveTechniques2020, debeaudrapTechniquesReducePi2020, kissingerReducingTcountZXcalculus2020}, quantum error correction~\cite{debeaudrapZXCalculusLanguage2020, kissingerPhasefreeZXDiagrams2022, khesinGraphicalQuantumCliffordencoder2023}, measurement-based quantum computation~\cite{coeckeInteractingQuantumObservables2007, duncanGraphStatesNecessity2009, kissingerUniversalMBQCGeneralised2019}, quantum natural language processing~\cite{coeckeFoundationsNearTermQuantum2020, meichanetzidisQuantumNaturalLanguage2021}, quantum simulation~\cite{kissingerClassicalSimulationQuantum2022}, quantum foundations~\cite{coeckePhaseGroupsOrigin2011, backensCompleteGraphicalCalculus2016}, cognition~\cite{signorelliCompositionalModelConsciousness2021}, and quantum education~\cite{coeckeQuantumPicturesExperiment2022, coeckeQuantumPictures2022}.
More detailed accounts on the ZX-calculus can be found in~\cite{coeckePicturingQuantumProcesses2017, vandeweteringZXcalculusWorkingQuantum2020} and a more extensive overview of its applications in~\cite{coeckeKindergardenQuantumMechanics2022}.

On the other hand, another graphical calculus, the ZW-calculus~\cite{coeckeThreeQubitEntanglement2011, hadzihasanovicDiagrammaticAxiomatisationQubit2015} is useful for studying multi-partite entanglement~\cite{coeckeThreeQubitEntanglement2011}, for describing fermionic quantum computing~\cite{ngDiagrammaticCalculusFermionic2019}, interactions in quantum field theory~\cite{shaikhCategoricalSemanticsFeynman2022}, and linear optical quantum computing~\cite{defeliceQuantumLinearOptics2022}.
These calculi have different applications because the ZW-calculus naturally expresses the summation of diagrams~\cite{coeckeGHZWcalculusContains2011}, while the ZX-calculus is more natural for gate-based and measurement-based quantum computing~\cite{coeckeInteractingQuantumObservables2007}.

In the \hbox{\emph{ZXW-calculus}}, we unify the ZX and ZW calculi in a single framework (\cref{fig:zxw}), by combining the generators and rules of both, and we also introduce new rules involving the interaction of the X- and W- generators.
This unified calculus allows us to utilise techniques native to both of these calculi, and get the best of both worlds.
This was not attempted until now as it was expected that the X- and W- generators would not play well together, but that turns out not to be the case.
\begin{figure}
  \input{figures/ZXW-figure-one-column.tex}
  \caption{Example applications of the ZXW-calculus and its relation to the ZX and ZW calculi.}
  \label{fig:zxw}
\end{figure}
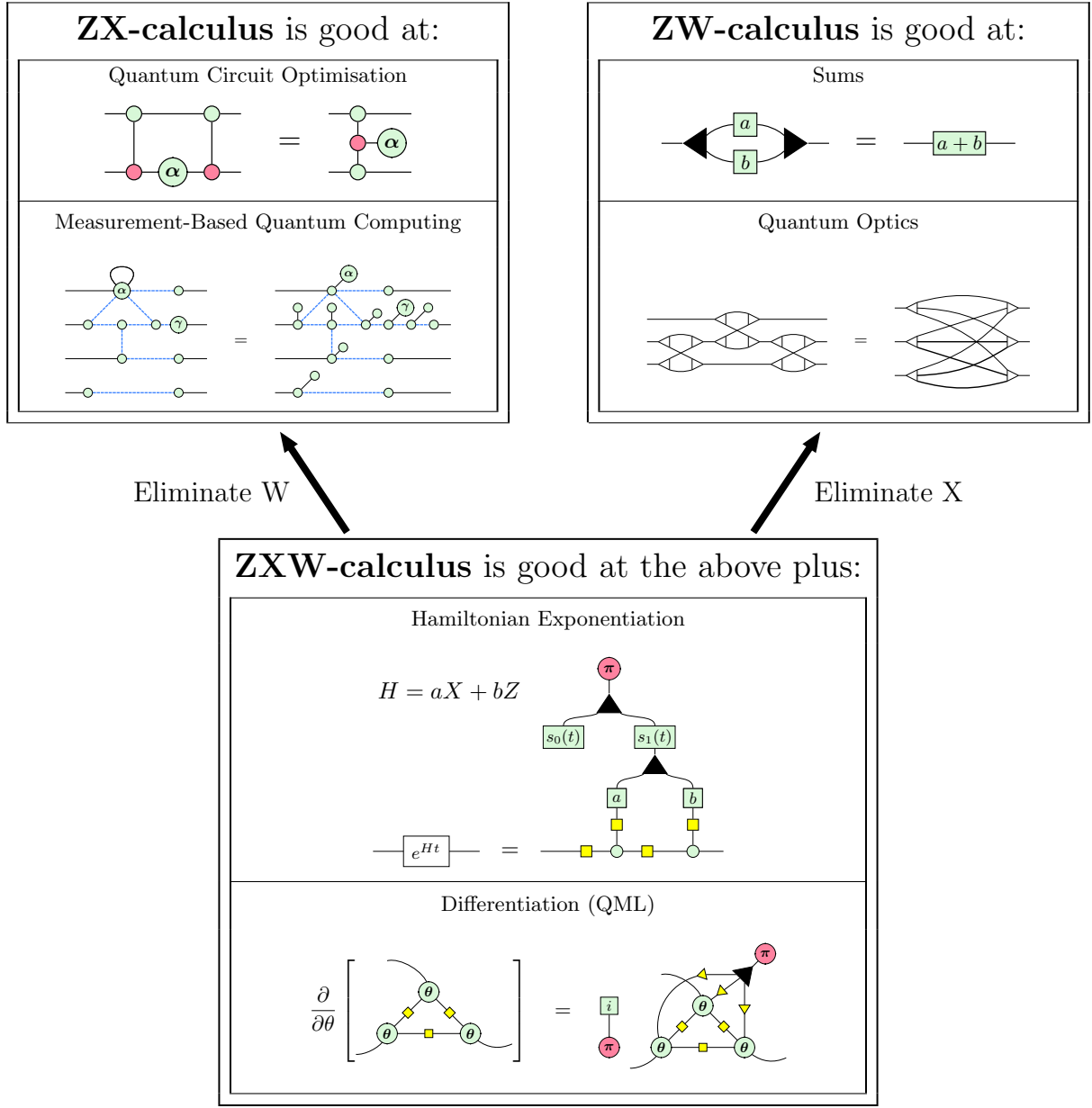

A graphical language for quantum computing, or any graphical language for that matter, typically consists of two parts, generators and rewrite rules.
There are three essential properties of a robust graphical language for quantum computing:
\begin{enumerate}
  \item \emph{universality}: every linear map of type $2^n \to 2^m$ can be expressed in the language,
  \item \emph{soundness}: the rules of the calculus respect multilinear algebra,\footnote{That is, linear algebra with tensor products.} and
  \item \emph{completeness}: any equation involving linear maps that is derivable in multilinear algebra, can also be derived in the graphical language by rewriting.
\end{enumerate}
Of these, completeness is the most difficult to prove.

\subsection{Related work}

\subsubsection{History of completeness results}

Although the qubit ZX-calculus was first formulated in 2007~\cite{coeckeInteractingQuantumObservables2007}, the crucial feature of completeness for all linear maps on qubits was not proved until 2017~\cite{ngUniversalCompletionZXcalculus2017}.
This moreover went in several stages of increasing the fragment for which completeness had been obtained~\cite{backensZXcalculusCompleteStabilizer2014, backensZXcalculusCompleteSinglequbit2014, jeandelCompleteAxiomatisationZXCalculus2018, hadzihasanovicTwoCompleteAxiomatisations2018}.
It required adding a number of rules following several observations of incompleteness~\cite{coeckeThreeQubitEntanglement2011,  duncanGraphStatesNecessity2009, dewittZXcalculusIncompleteQuantum2014}.
These were subsequently simplified to only one additional rule~\cite{coeckeZXRules2QubitClifford2018, vilmartNearMinimalAxiomatisationZXCalculus2019}.
It was a highly non-trivial process but now each rule is concise and intuitive.

In fact, the first completeness results for the whole of the qubit ZX-calculus~\cite{ngUniversalCompletionZXcalculus2017, jeandelDiagrammaticReasoningClifford2018} was achieved through transferring the completeness proof of the qubit ZW-calculus~\cite{hadzihasanovicDiagrammaticAxiomatisationQubit2015}, which was obtained earlier in 2015.
While this completeness result for the ZW-calculus was an incredible achievement, what made it feasible to obtain it in one go is the fact that the W-generator enables one to easily express sums, and hence easily express the matrices of linear algebra.

\subsubsection{Higher-dimensional quantum graphical calculi}

In contrast to the conventional 2-level qubits, \emph{qudits} -- which are multi-level units of quantum information -- have gained attention as a topic of study and commercial interest.
It provides a larger state space for processing information, and thus has simplified quantum algorithms~\cite{wangQuditsHighDimensionalQuantum2020}, improved quantum error correction~\cite{campbellEnhancedFaulttolerantQuantum2014}, and enables native quantum simulation of certain physical systems~\cite{banulsSimulatingLatticeGauge2020}.
In recent years, qudit-based quantum computation has been realised with ion traps~\cite{ringbauerUniversalQuditQuantum2022, hrmoNativeQuditEntanglement2022}, photonic devices~\cite{chiProgrammableQuditbasedQuantum2022}, and superconducting devices~\cite{blokQuantumInformationScrambling2021, yeCircuitQEDSinglestep2018, yurtalanImplementationWalshHadamardGate2020, hillRealizationArbitraryDoublycontrolled2021}.

In addition to qubit graphical calculi, there is a growing body of literature on the creation and development of higher-dimensional quantum graphical calculi.
The first qutrit (3-level qudit) ZX-calculus~\cite{wangQutritDichromaticCalculus2014} and the first qudit ZX-calculus~\cite{ranchinDepictingQuditQuantum2014} were introduced in 2014.
Universality of qudit ZX-calculus for all qudit linear maps was shown in 2014~\cite{wangQutritDichromaticCalculus2014}.
Through a translation between qudit ZX-calculus and (non-anyonic) qudit ZW-calculus, qudit ZW-calculus was also shown to be universal~\cite{wangNonanyonicQuditZWcalculus2021}.
Meanwhile, completeness of the qutrit ZX-calculus for the stabilizer fragment was proved in 2018~\cite{wangQutritZXcalculusComplete2018}, and for the stabilizer fragment for all prime dimension qudits in 2022~\cite{boothCompleteZXcalculiStabiliser2022}.

Recently, these qudit graphical calculi have been applied to achieve new results in qudit circuit synthesis.
By leveraging the qutrit ZX-calculus, novel qutrit constructions have been elucidated from their qubit ZX-calculus counterparts, for instance, discovering a qutrit unitary that implements the qubit CCZ gate using only three single-qudit non-Clifford gates~\cite{vandeweteringPhaseGadgetCompilation2022}.
Another recent example is the application of qutrit ZX-calculus to synthesizing two-qutrit controlled phase gates on superconducting qutrits~\cite{caoQutritZXCalculusSuperconducting2022}.

\subsubsection{Application of ZXW-calculus}

The ZXW-calculus has been used to realise the differentiation and integration of ZX-diagrams which has been used to analyse barren plateaus in context of quantum machine learning~\cite{wangDifferentiatingIntegratingZX2022}.
It is also possible to express Hamiltonian summation and exponentiation in the ZXW-calculus that sets up the framework to solve problems in quantum chemistry and condensed matter physics~\cite{shaikhHowSumExponentiate2022}.

\subsection{Our contributions}\label{subsec:our-contributions}

As a first step towards achieving our main objective, we introduce the ZXW-calculus for arbitrary finite dimensions in \cref{sec:zxw_calculus}.
Furthermore, we present its unique normal form and give a set of rewrite rules for the calculus.
In \cref{sec:completeness}, we prove that our set of rules is complete for all finite dimensions.
Lastly, we discuss our results and its applications at present and in the foreseeable future in \cref{sec:future}.

With the ZXW-calculus, here we achieve the first completeness result for universal quantum computation in arbitrary finite dimensions.
We give a procedure to rewrite one diagram to another whenever they evaluate to the same linear map.
Furthermore, we prove that any diagram can be transformed into its unique normal form diagram, originally given in~\cite{wangQufiniteZXcalculusUnified2022}.
This is accomplished using only a finite set of diagrammatic rules, which are presented in \cref{subsec:rule-set}.
This makes the ZXW-calculus a powerful tool for reasoning about qudit quantum computing.

%% file: figures/ZXW-figure-one-column.tex
\begingroup

\center
\tikzstyle{hadamard edge}=[-, dashed, dash pattern=on 2pt off 0.5pt, thick, draw={rgb,255: red,68; green,136; blue,255}]

\makeatletter
\newsavebox\saved@arstrutbox
\newcommand*{\setarstrut}[1]{%
  \noalign{%
    \begingroup
    \global\setbox\saved@arstrutbox\copy\@arstrutbox
    #1%
    \global\setbox\@arstrutbox\hbox{%
      \vrule \@height\arraystretch\ht\strutbox
      \@depth\arraystretch \dp\strutbox
      \@width\z@
    }%
    \endgroup
  }%
}
\newcommand*{\restorearstrut}{%
  \noalign{%
    \global\setbox\@arstrutbox\copy\saved@arstrutbox
  }%
}
\makeatother

\tabulinestyle {.75pt}
\tabulinesep=2mm
\setlength{\tabcolsep}{2pt}
\newcommand{\lastlinenewstrech}{.35}

\begin{tabu}
  to \textwidth {|c X[c] c|p{1cm}|c X[c] c|}
  \cline{1-3} \cline{5-7}
  & \Large{\textbf{ZX-calculus} is good at:}                                             & & & & \Large{\textbf{ZW-calculus} is good at:}                 & \\
  \cline{2-2} \cline{6-6} \tabuphantomline
  & \multicolumn{1}{|c|}{\footnotesize{Quantum Circuit Optimisation}}                    & & & & \multicolumn{1}{|c|}{\footnotesize{Sums}}                                      & \\
  & \multicolumn{1}{|c|}{\scalebox{1.15}{\input{introduction/optimisation.tikz}}}           & & & & \multicolumn{1}{|c|}{\scalebox{1.1}{\input{introduction/summation.tikz}}}         & \\
  \cline{2-2} \cline{6-6}
  & \multicolumn{1}{|c|}{\footnotesize{Measurement-Based Quantum Computing}}             & & & & \multicolumn{1}{|c|}{\footnotesize{Quantum Optics}}          & \\
  & \multicolumn{1}{|c|}{\scalebox{.67}{\input{introduction/graph-like-to-MBQC-form.tikz}}}             & & & & \multicolumn{1}{|c|}{\scalebox{.665}{\input{introduction/path-normal-form.tikz}}} & \\
  \cline{2-2} \cline{6-6}
  \setarstrut{\renewcommand*{\arraystretch}{\lastlinenewstrech}}
  &                                                                                      & & & &                                                                                & \\
  \cline{1-3} \cline{5-7}
  \restorearstrut
\end{tabu}
\usetikzlibrary{arrows.meta}
\begin{tikzpicture}[my triangle/.style={-{Triangle[width=\the\dimexpr2\pgflinewidth,length=\the\dimexpr3\pgflinewidth]}}]
  \draw[line width=3pt,my triangle](0,0) -- (-1, 1.5);
  \node [style=none] at (-2, 0.6) {\large Eliminate W};
  \draw[line width=3pt,my triangle](6,0) -- (7, 1.5);
  \node [style=none] at (8, 0.6) {\large Eliminate X};
\end{tikzpicture} \\
\begin{tabu}
  to \textwidth {|c c c|}
  \tabucline-
  & \Large{\textbf{ZXW-calculus} is good at the above plus:}        & \\
  \cline{2-2}
  & \multicolumn{1}{|c|}{\footnotesize{Hamiltonian Exponentiation}} & \\
  & \multicolumn{1}{|c|}{\scalebox{.9}{
    \input{introduction/ExtractionHamiltonianExponential.tikz}
  }} & \\
  \cline{2-2}
  & \multicolumn{1}{|c|}{\footnotesize{Differentiation (QML)}}      & \\
  & \multicolumn{1}{|c|}{\scalebox{.82}{
    $\dfrac{\partial}{\partial\theta}
    \left[
      \input{introduction/graphstate.tikz}
      \right]
    \quad = \quad
    \input{introduction/derivativegraphstate.tikz}$
  }} & \\
  \cline{2-2}
  \setarstrut{\renewcommand*{\arraystretch}{\lastlinenewstrech}}
  &                                                                 & \\
  \tabucline-
  \restorearstrut
\end{tabu}

\endgroup

%% file: zxw-calculus.tex
In this section, we give an introduction to the general notion of a graphical calculus, and specifically the ZXW-calculus for qudits, its generators, their standard interpretation, and a complete set of rules.

\subsection{A general graphical calculus}\label{subsec:graphical-calculus}

A graphical calculus is a mathematical language for expressing and reasoning about processes using graphical representations (i.e.\@ diagrams) whilst maintaining the rigour of having an explicit mathematical interpretation for each.
By working with a graphical representation, we obtain higher expressivity, enabling a cleaner presentation of the underlying structure than if we were to rely solely on conventionally used symbolic notation.
The diagrams of a graphical calculus are composed of boxes and wires, which can be manipulated according to a set of equalities called graphical rewrite rules.
In this section, we explain the basic idea of a general graphical calculus.
For more details, a thorough discussion on the topic is presented in~\cite{coeckeGeneralisedCompositionalTheories2016}.

The basic building blocks of a graphical calculus are called \emph{generators}:
\[
  \input{graphical-calculi/generators.tikz}
\]
We can construct more complex diagrams from these generators by composing them sequentially or in parallel.
\emph{Sequential composition} (denoted by $\circ$) of two diagrams means connecting the outputs of one diagram to the inputs of the other:
\[
  \input{graphical-calculi/generalboxf.tikz}\ \circ \ \input{graphical-calculi/generalboxg.tikz}
  \quad \coloneqq \quad
  \input{graphical-calculi/sequelcomposition.tikz}
\]
\emph{Parallel composition} (denoted by $\otimes$) of two diagrams is represented by placing the two diagrams next to each other:
\[
  \input{graphical-calculi/generalboxf.tikz} \ \otimes \ \input{graphical-calculi/generalboxg.tikz}
  \quad \coloneqq \quad
  \input{graphical-calculi/parallelcomposition.tikz}
\]

In addition to constructing diagrams, we can also reason about them by transforming one into another.
This can be done by replacing any sub-diagram within a complex diagram according to a set of graphical rewrite rules, that is,
\[
  \input{graphical-calculi/rewritedm1LHS.tikz}
  \quad = \quad
  \input{graphical-calculi/rewritedm1RHS.tikz}
\]
given that
\[
  \input{graphical-calculi/rewritedm2LHS.tikz}
  \quad = \quad
  \input{graphical-calculi/rewritedm2RHS.tikz}\,.
\]
In particular, a lot of well-known graphical calculi, including the graphical calculus used in this paper, obey the following structural rules that allow us to move diagrams around freely:
\begin{gather*}
  \input{graphical-calculi/compactstructure_1.tikz}
  \qquad \qquad
  \input{graphical-calculi/compactstructure_2.tikz}
  \qquad \qquad
  \input{graphical-calculi/compactstructure_3.tikz}
\end{gather*}

Lastly, to relate a graphical calculus to a system of interest, one needs to give an interpretation to its diagrams such that the sequential and parallel compositions are preserved.
In the case of quantum computation, the interpretation of diagrams is matrices, the sequential composition is matrix multiplication, and parallel composition is the tensor product.

\subsection{The ZXW-calculus}\label{subsec:the-zxw-calculus}

In this section, we explore a specific graphical calculus called the ZXW-calculus.
Throughout the paper, $d$ is the dimension of the calculus, an arbitrary integer greater than one.
We define the interpretation functor $\interp{\cdot}: \textbf{ZXW}_d \to \mathbf{Vect}_d$, where $\textbf{ZXW}_d$ is the category of qudit ZXW diagrams and $\mathbf{Vect}_d$ is the category of qudits.
The functor $\interp{\cdot}$ is defined inductively in usual way for both compositions, that is, sequential composition is matrix multiplication and parallel composition is the tensor product.
The interpretation of the generators are given below.

\subsubsection{Generators of ZXW-calculus and their interpretation}\label{subsec:generators_n_interpretation}

The generators of the ZXW-calculus together with their standard interpretation $\interp{\cdot}$ are:
\begin{itemize}
  \item The \emph{Z box},
  \[
    \input{qudit-generators/generalgreenspiderqdit2.tikz}
    \quad \overset{\interp{\cdot}}{\longmapsto} \quad
    \sum_{j=0}^{d-1}a_j\ket{j}^{\otimes m}\bra{j}^{\otimes n},
  \]
  where $a_0 = 1$ and $\overrightarrow{a}=(a_1, \cdots, a_{d-1})$ is an arbitrary complex vector, and we take the indices modulo $d$, that is, $a_j = a_{j\ \mathrm{mod}\ d}$ for $j \in \Z$.
  Having a $d - 1$ dimensional vector is the direct generalisation of the qubit case where the parameter of a spider is a single number.

  \item The \emph{Hadamard box},
  \[
    \input{qudit-generators/HadaDecomSingleslt.tikz}
    \quad \overset{\interp{\cdot}}{\longmapsto} \quad
    \frac{1}{\sqrt{d}}\sum_{k, j=0}^{d-1}\omega^{jk}\ket{j}\bra{k},
  \]
  where $\omega = e^{i\frac{2\pi}{d}}$ is the $d$-th root of unity.
  In the qudit case the Hadamard box is not self-adjoint anymore as in the qubit case, so two Hadamard boxes do not equal the identity.

  \item The \emph{W node},
  \[
    \input{qudit-generators/w1to2.tikz}
    \quad \overset{\interp{\cdot}}{\longmapsto} \quad
    \ket{00}\bra{0}+\sum_{i=1}^{d-1}(\ket{0i}+\ket{i0})\bra{i}.
  \]
  \item The \emph{swap},
  \[
    \input{qudit-generators/swap.tikz}
    \quad \overset{\interp{\cdot}}{\longmapsto} \quad
    \sum_{i, j=0}^{d-1}\ket{ji}\bra{ij}.
  \]
  \item Lastly, the \emph{identity},
  \[
    \input{qudit-generators/Id.tikz}
    \quad \overset{\interp{\cdot}}{\longmapsto} \quad
    I_d=\sum_{j=0}^{d-1}\ket{j}\bra{j}.
  \]
\end{itemize}

\subsubsection*{Notations}

For convenience, we introduce the following notation which will be used throughout the paper:
\begin{itemize}
  \item The original green circle spider~\cite{coeckeInteractingQuantumObservables2011, ranchinDepictingQuditQuantum2014} can be defined using the Z box:
  \begin{gather*}
    \input{definitions/circlegspiders_1.tikz}
    \qquad \qquad
    \input{definitions/circlegspiders_2.tikz}
  \end{gather*}
  where
  $\overrightarrow{1} = \overbrace{(1,\cdots,1)}^{d-1}$,\,
  $\overrightarrow{\alpha} = (\alpha_1, \cdots, \alpha_{d-1})$,\,
  $e^{i\overrightarrow{\alpha}}=(e^{i\alpha_1}, \cdots, e^{i\alpha_{d-1}})$,\, and
  $\alpha_i \in [0, 2\pi)$.
  Thus, the interpretation of the green circle spider is as follows:
  \[
    \input{definitions/circlegspider.tikz}
    \quad \overset{\interp{\cdot}}{\longmapsto} \quad
    \sum_{j=0}^{d-1} e^{i \alpha_j} \ket{j}^{\otimes m}\bra{j}^{\otimes n},
    \qquad \text{where} \quad
    \alpha_0 \coloneqq 0.
  \]

  \item When the first $d-2$ components of the parameter vector in a Z box are all zeros, we simply label the box with its last component:
  \[
    \input{definitions/lastditgbox2.tikz}\ ,
  \]
  where $x \in \mathbb C$.

  \item The qudit version of the Bell state is $\ket{00} + \ket{11} + \ket{22} + \cdots + \ket{(d-1) (d-1)}$.
  In the ZXW-calculus, the Bell state and its transpose can be defined as, respectively:
  \begin{gather}
    \input{definitions/compactstructures_1.tikz}
    \qquad \qquad
    \input{definitions/compactstructures_2.tikz}
    \tag{S3}\label{rule:S3}\refstepcounter{equation}
  \end{gather}
  We refer to these diagrams as caps and cups, and structures equipped with them are called compact structures~\cite{coeckeInteractingQuantumObservables2011}.
  Unlike in the qubit case, the green and red caps and cups do not coincide:
  \[
    \input{definitions/xbellneqzbell.tikz}
  \]

  \item The Hadamard box is the generalised Hadamard gate for qudits.
  We define its inverse, the yellow $H^\dagger$ box as follows:
  \begin{gather}
    \input{definitions/h-dagger.tikz}
    \tag{H$\null^\dagger$}\label{rule:HDagger}\refstepcounter{equation}
  \end{gather}
  It is proved in the appendix (\cref{hilm}) that the above is well defined.

  \item The yellow $H$ and $H^\dagger$ boxes can be used to define the X-spider~\cite{ranchinDepictingQuditQuantum2014};
  however, we only define the normalised version of the X-spider called the pink spider, for some particular phase vectors.
  These two versions of spiders are the same up to some variational scalars depending on the number of inputs and outputs:\\
  \begin{minipage}{.6\linewidth}
    \begin{gather}
      \input{definitions/pinkspiders.tikz}
      \tag{HZ}\label{rule:HZ}\refstepcounter{equation}
    \end{gather}
  \end{minipage}
  \begin{minipage}{.4\linewidth}
    \[
      \input{definitions/pinkspiders-phasefree.tikz}
    \]
  \end{minipage}
  where, for $0 \leq j \leq d-1$, $K_j$ is the $j$-th column (or row) of the quantum Fourier transform matrix (or the Hadamard matrix), that is,
  \[
    K_j=\left(j\frac{2\pi}{d}, 2j\frac{2\pi}{d}, \cdots, (d-1)j\frac{2\pi}{d}\right),
  \]
  and $j$ can be taken modulo $d$;
  furthermore, $u_{m,n} = d^{\frac{m+n-2}{2}}-1$, hence the green box \input{definitions/umngbox.tikz} represents the scalar $d^{\frac{m+n-2}{2}}$.
  The $K_j$ vectors phases are similar to the qubit $0$ and $\pi$ phases in a sense that these correspond to the phases of Pauli operators.
  In the qudit case, there are exactly $d$ such phases.

  It is useful to note that the interpretation of the pink spider for each $K_j$ phase is given as:
  \[
    \input{qudit-generators/quditrspiderclassicnm.tikz}
    \quad \overset{\interp{\cdot}}{\longmapsto} \hspace{-.5cm}
    \sum_{\substack{
      0 \leq i_1, \cdots, i_m,  j_1, \cdots, j_n \leq d-1 \\
      i_1+\cdots+ i_m+j \equiv j_1+\cdots +j_n \Mod{d}}
    } \hspace{-1.75cm} \ket{i_1, \cdots, i_m}\bra{j_1, \cdots, j_n}.
  \]

  \item We use a yellow $D$ box to denote the \emph{dualiser} as defined in~\cite{coeckeInteractingQuantumObservables2011}:
  \begin{gather}
    \input{definitions/dualiser.tikz}
    \tag{Du}\label{rule:Du}\refstepcounter{equation}
  \end{gather}
  with interpretation
  \[
    \input{definitions/dualiser-d.tikz}
    \quad \overset{\interp{\cdot}}{\longmapsto} \quad
    \sum_{i = 0}^{d \minu 1} \ket{i} \bra{-i}.
  \]
  where $\bra{-i}$ is defined modulo $d$.
  In the qubit case, the dualiser is just the identity, but generally it is an involution (see \cref{dboxsquarelm}).

  \item The general $W$ node \scalebox{.75}{\input{definitions/w1ton.tikz}} and its transpose \scalebox{.75}{\input{definitions/wnto1.tikz}} are defined as:
  \begin{gather}
    \input{definitions/mlegsblackspider.tikz}
    \qquad\qquad
    \input{definitions/mlegsdownblackspider.tikz}
    \tag{WN}\label{rule:WN}\refstepcounter{equation}
  \end{gather}
  with interpretation
  \begin{align*}
    \input{definitions/w1ton.tikz}
    \quad \overset{\interp{\cdot}}{\longmapsto} \quad
    &\ket{0\cdots0}\bra{0}\\
    + &\sum_{i=1}^{d - 1}(\ket{i0\cdots 00}+\cdots +\ket{00\cdots 0i})\bra{i}.
  \end{align*}
  Note that a W node with a single leg is the identity.
  \begin{remark}
    The W node is named after the \emph{N-qudit standard W state}~\cite{wuNquditSLOCCEquivalent2020}:
    \[
      \ket{W_{N}^d}=\frac{1}{\sqrt{N(d-1)}}\sum_{i=1}^{d-1}(\ket{i0\cdots 00}+\cdots +\ket{00\cdots 0i})
    \]
    which, up to a scalar, can be represented in the ZXW-calculus as:
    \begin{equation}
      \input{definitions/wstate.tikz}
      \label{eq:wstate}
    \end{equation}
  \end{remark}

  \item It is useful to define the yellow triangle and its inverse in terms of the $W$ node and green nodes:
  \begin{gather}
    \input{definitions/trianlge-def.tikz}
    \qquad \qquad
    \input{definitions/trianlge-inv-def.tikz}
    \tag{YT}\label{rule:YT}\refstepcounter{equation}
  \end{gather}
  where $\overrightarrow{\minu 1} = \overbrace{(\minu 1,\,\cdots,\,\minu 1)}^{d - 1}$.
  Furthermore, their interpretation is given as:
  \begin{gather*}
    \input{definitions/trianlge.tikz}
    \quad \overset{\interp{\cdot}}{\longmapsto} \quad
    I_d + \sum_{i=1}^{d  \minu 1} \ket{0} \bra{i}
    \qquad \qquad
    \input{definitions/trianlge-inv.tikz}
    \quad \overset{\interp{\cdot}}{\longmapsto} \quad
    I_d - \sum_{i=1}^{d  \minu 1} \ket{0} \bra{i}
  \end{gather*}

  \item A multiplier~\cite{bonchiInteractingHopfAlgebras2017, caretteSZXCalculusScalableGraphical2019, boothCompleteZXcalculiStabiliser2022} labelled by $m$ indicates the number of connections between green and pink nodes.
  Unlike in the qubit case, a green and a red spider can be connected with more than one wire.
  In fact, the Hopf law generalises to $d$ connections (see \cref{hopfditlm}) for a red and green spider to disconnect, so $m$ can be labeled modulo $d$:
  \begin{gather}
    \input{definitions/multiplier.tikz}
    \qquad \qquad
    \input{definitions/multiplier-t.tikz}
    \qquad \qquad
    \input{lemmas/multipliers/multiplier-mod.tikz}
    \tag{Mu}\label{rule:Mu}\refstepcounter{equation}
  \end{gather}

  \item The multipliers interacting together with the Z, X and W nodes play an important role in the ZXW-calculus.
  For brevity, we use the notations $V$ and $M$ to mean:\\
  \begin{minipage}{.5\linewidth}
    \begin{gather}
      \input{definitions/v-def.tikz}
      \tag{VB}\label{rule:VB}\refstepcounter{equation}
    \end{gather}
  \end{minipage}
  \begin{minipage}{.5\linewidth}
    \begin{gather}
      \input{definitions/m-def.tikz}
      \tag{MB}\label{rule:MB}\refstepcounter{equation}
    \end{gather}
  \end{minipage}
  When $d=2$, the $V$ box is the identity because the single-legged W node is also the identity according to the definition of multi-legged W node~\eqref{rule:WN}.
  The interpretation of the $V$ box is given as:
  \[
    \input{definitions/v-box.tikz}
    \quad \overset{\interp{\cdot}}{\longmapsto} \quad
    \ket{0}\bra{0} + \sum_{i=1}^{d \minu 1} \ket{i}\bra{-1}
  \]
  Furthermore, the $M$ box can be expressed in terms of the $V$ box as follows:
  \[
    \input{definitions/m-box.tikz}
  \]
\end{itemize}
\begin{remark}
  Our the language gives up on the property called \enquote{only topology matter}~\cite{caretteWhenOnlyTopology2021}, that is, our diagrams cannot be seen as open graphs.
  This because neither the W node nor the pink spider are flexsymmetric.
  We can define a flexsymmetric version of the pink spider by replacing the $H^\dagger$ boxes with $H$ boxes in \cref{rule:HZ}.
  However, such a definition results in a more complicated a fusion rule between pink spiders.
  Furthermore, defining a well-behaving, flexsymmetric alternative to the W node is a complicated task that we leave for future work.
\end{remark}
\begin{remark}
  When defining the language, we have the freedom of switching between some of the generators and notations (e.g.\@ the X-spider and the Hadamard box).
  However, we specifically chose this set of generators because it results in an easy and clean proof of completeness.
\end{remark}

\subsubsection{Complete rule set}\label{subsec:rule-set}

\input{axioms-table.tex}

\subsubsection{Normal form}\label{subsec:normal-form}

Given an arbitrary complex vector $\overrightarrow{a} = \left( a_0, a_1, \cdots, a_{d^{m}-1} \right)^T$ of dimension $d^m$,
there is a normal form, such that
\[
  \input{definitions/normalformdit.tikz}
  \quad \overset{\interp{\cdot}}{\longmapsto} \quad
  \begin{pmatrix}
    a_0       \\
    \vdots    \\
    a_i       \\
    \vdots    \\
    a_{d^m-1} \\
  \end{pmatrix}
\]
where the weights of the multipliers from the $i$-th Z box $a_i$ is the $d$-ary expansion of the given index, that is,
\[
  \begin{pmatrix}
    a_0       \\
    \vdots    \\
    a_i       \\
    \vdots    \\
    a_{d^m-1} \\
  \end{pmatrix}
  \quad = \quad
  \sum_{j=0}^{d^m-1} a_j \ket{e_{m \minu 1, j} \,\cdots\, e_{s, j} \,\cdots\, e_{0, j}}\ ,
\]
where $0 \leq e_{k, j} \leq d-1$.
This normal form was originally designed for the algebraic qudit ZX-calculus~\cite{wangQufiniteZXcalculusUnified2022}.
It is also structurally similar to the normal form for qudit ZW-calculus~\cite{hadzihasanovicAlgebraEntanglementGeometry2017}.

\begin{remark}
  The qudit normal form is unique if we fix the order of green boxes based on the weights of the multipliers it is connected to.
\end{remark}

\begin{definition}
  A normal form is in its \emph{unique form} if the Z boxes are ordered according to the lexicographical ordering of the multiplier connections.
\end{definition}

\begin{restatable}{lemma}{nftounf}
  \label{lem:nf2unf}
  Any normal form can be transformed into its unique form.
\end{restatable}

\begin{remark}
  If two outputs of a unique form are swapped, it remains a normal form, but not the unique one.
  Then, this normal form can be rewritten into its unique form using \cref{lem:nf2unf}.
\end{remark}

\begin{remark}
  The normal form of a scalar diagram is of the form \input{definitions/s-scalar.tikz}.
\end{remark}

%% file: figures/qudit-generators/generalgreenspiderqdit2.tikz
\begin{tikzpicture}
	\begin{pgfonlayer}{nodelayer}
		\node [style=none] (0) at (1.625, -1.875) {};
		\node [style=none] (2) at (-1.625, -1.875) {};
		\node [style=none] (5) at (1.625, 1.875) {};
		\node [style=none] (9) at (-1.625, 1.875) {};
		\node [style=none] (13) at (0, -2.375) {{\scriptsize $m$}};
		\node [style=none] (14) at (0, 2.375) {{\scriptsize $n$}};
		\node [style=none] (15) at (0.5, 1.25) {$\cdots$};
		\node [style=none] (16) at (0.5, -1.25) {$\cdots$};
		\node [style=none] (17) at (-1.5, -1.5) {};
		\node [style=none] (18) at (-0.5, -1.5) {};
		\node [style=none] (19) at (1.5, -1.5) {};
		\node [style=gbox] (20) at (0, 0) {$\overrightarrow{a}$};
		\node [style=none] (21) at (-1.5, 1.5) {};
		\node [style=none] (22) at (-0.5, 1.5) {};
		\node [style=none] (23) at (1.5, 1.5) {};
	\end{pgfonlayer}
	\begin{pgfonlayer}{edgelayer}
		\draw [style=braceedge] (0.center) to (2.center);
		\draw [style=braceedge] (9.center) to (5.center);
		\draw [style=none, bend left=15] (17.center) to (20);
		\draw [style=none, bend left=15] (18.center) to (20);
		\draw [style=none, bend left=15] (20) to (21.center);
		\draw [style=none, bend right=15] (19.center) to (20);
		\draw [style=none, bend right=15] (20) to (23.center);
		\draw [style=none, bend left=15] (20) to (22.center);
	\end{pgfonlayer}
\end{tikzpicture}

%% file: figures/qudit-generators/HadaDecomSingleslt.tikz
\begin{tikzpicture}
	\begin{pgfonlayer}{nodelayer}
		\node [style=hbox] (0) at (0, 0) {$H$};
		\node [style=none] (1) at (0, -1) {};
		\node [style=none] (2) at (0, 1) {};
	\end{pgfonlayer}
	\begin{pgfonlayer}{edgelayer}
		\draw (2.center) to (0);
		\draw (1.center) to (0);
	\end{pgfonlayer}
\end{tikzpicture}

%% file: figures/qudit-generators/w1to2.tikz
\begin{tikzpicture}
	\begin{pgfonlayer}{nodelayer}
		\node [style=none] (1) at (0, 1) {};
		\node [style=uw] (4) at (0, 0) {};
		\node [style=none] (6) at (-0.625, -1) {};
		\node [style=none] (7) at (0.625, -1) {};
	\end{pgfonlayer}
	\begin{pgfonlayer}{edgelayer}
		\draw [in=90, out=-135, looseness=1.25] (4) to (6.center);
		\draw [in=90, out=-45, looseness=1.25] (4) to (7.center);
		\draw (1.center) to (4);
	\end{pgfonlayer}
\end{tikzpicture}

%% file: figures/qudit-generators/swap.tikz
\begin{tikzpicture}
	\begin{pgfonlayer}{nodelayer}
		\node [style=none] (0) at (0.0, 0.0) {};
		\node [style=none] (1) at (0.6, 0.6) {};
		\node [style=none] (2) at (0.6, -0.6) {};
		\node [style=none] (a) at (-0.6, 0.6) {};
		\node [style=none] (b) at (-0.6, -0.6) {};
		\node [style=none] (3) at (0.6, -1.0) {};
		\node [style=none] (4) at (0.6, 1.0) {};
	\end{pgfonlayer}
	\begin{pgfonlayer}{edgelayer}
		\draw[bend left=35] (1.center) to (0.center);
		\draw[bend right=35] (2.center) to (0.center);
		\draw[bend right=35] (a.center) to (0.center);
		\draw[bend left=35] (b.center) to (0.center);
	\end{pgfonlayer}
\end{tikzpicture}

%% file: figures/qudit-generators/Id.tikz
\begin{tikzpicture}
	\begin{pgfonlayer}{nodelayer}
		\node [style=none] (1) at (1.0, 0.6) {};
		\node [style=none] (2) at (1.0, -0.6) {};
		\node [style=none] (3) at (1.0, -1.0) {};
		\node [style=none] (4) at (1.0, 1.0) {};
	\end{pgfonlayer}
	\begin{pgfonlayer}{edgelayer}
		\draw (1.center) to (2.center);
	\end{pgfonlayer}
\end{tikzpicture}

%% file: figures/definitions/umngbox.tikz
\begin{tikzpicture}
	\begin{pgfonlayer}{nodelayer}
		\node [style=gbox] (0) at (0, 0) {$u_{m,n}$};
	\end{pgfonlayer}
\end{tikzpicture}

%% file: figures/qudit-generators/quditrspiderclassicnm.tikz
\begin{tikzpicture}
	\begin{pgfonlayer}{nodelayer}
		\node [style=none] (0) at (0, 1.25) {$\cdots$};
		\node [style={rn_phase}] (1) at (0, 0) {$K_j$ };
		\node [style=none] (2) at (1, -1.25) {};
		\node [style=none] (3) at (-1, 1.25) {};
		\node [style=none] (4) at (0, -1.25) {$\cdots$};
		\node [style=none] (5) at (-1, -1.25) {};
		\node [style=none] (6) at (1, 1.25) {};
		\node [style=none] (7) at (-1.125, 1.875) {};
		\node [style=none] (8) at (1.125, 1.875) {};
		\node [style=none] (9) at (-1.125, -1.875) {};
		\node [style=none] (10) at (1.125, -1.875) {};
		\node [style=none] (11) at (0, 2.375) {{\scriptsize $n$}};
		\node [style=none] (12) at (0, -2.375) {{\scriptsize $m$}};
		\node [style=none] (13) at (1, 1.5) {};
		\node [style=none] (14) at (-1, 1.5) {};
		\node [style=none] (15) at (1, -1.5) {};
		\node [style=none] (16) at (-1, -1.5) {};
	\end{pgfonlayer}
	\begin{pgfonlayer}{edgelayer}
		\draw [in=-90, out=150] (1) to (3.center);
		\draw [in=-90, out=30] (1) to (6.center);
		\draw [in=90, out=-150] (1) to (5.center);
		\draw [in=90, out=-30] (1) to (2.center);
		\draw [style=braceedge] (10.center) to (9.center);
		\draw [style=braceedge] (7.center) to (8.center);
		\draw (13.center) to (6.center);
		\draw (14.center) to (3.center);
		\draw (16.center) to (5.center);
		\draw (15.center) to (2.center);
	\end{pgfonlayer}
\end{tikzpicture}

%% file: figures/definitions/dualiser-d.tikz
\begin{tikzpicture}
	\begin{pgfonlayer}{nodelayer}
		\node [style=hbox] (0) at (0, 0) {$D$};
		\node [style=none] (1) at (0, 1) {};
		\node [style=none] (2) at (0, -1) {};
	\end{pgfonlayer}
	\begin{pgfonlayer}{edgelayer}
		\draw (1.center) to (2.center);
	\end{pgfonlayer}
\end{tikzpicture}

%% file: figures/definitions/trianlge.tikz
\begin{tikzpicture}
	\begin{pgfonlayer}{nodelayer}
		\node [style=none] (0) at (0, 1) {};
		\node [style=none] (1) at (0, -1) {};
		\node [style=utriang] (2) at (0, 0) {};
	\end{pgfonlayer}
	\begin{pgfonlayer}{edgelayer}
		\draw (0.center) to (1.center);
	\end{pgfonlayer}
\end{tikzpicture}

%% file: figures/definitions/trianlge-inv.tikz
\begin{tikzpicture}
	\begin{pgfonlayer}{nodelayer}
		\node [style=none] (0) at (0, 1) {};
		\node [style=utriang] (1) at (0, 0) {};
		\node [style=none] (2) at (0, -1) {};
		\node [style=none] (3) at (0.375, 0.375) {\scriptsize $\null^{\minu 1}$};
	\end{pgfonlayer}
	\begin{pgfonlayer}{edgelayer}
		\draw (0.center) to (2.center);
	\end{pgfonlayer}
\end{tikzpicture}

%% file: figures/definitions/v-box.tikz
\begin{tikzpicture}
	\begin{pgfonlayer}{nodelayer}
		\node [style=none] (0) at (0, -1) {};
		\node [style=none] (1) at (0, 1) {};
		\node [style=box] (2) at (0, 0) {$V$};
	\end{pgfonlayer}
	\begin{pgfonlayer}{edgelayer}
		\draw (1.center) to (2);
		\draw (0.center) to (2);
	\end{pgfonlayer}
\end{tikzpicture}

%% file: axioms-table.tex
\begingroup 

\allowdisplaybreaks

In this section, we give a set of rewrite rules which we show to be complete for any finite dimensions.
In the ZXW-calculus, we not only get all the ZX and ZW rules, but also nice rules which involve the interaction of all three (Z, X, and W) algebras.
These additional rules and the flexibility to convert to some extent between X and W spiders capture intricacies that would otherwise be unwieldy to express if limited to solely ZX or solely ZW calculus\@.
Throughout the paper, we use the arbitrary complex vectors $\overrightarrow{a} = (a_1,\dotsc, a_{d-1})$ and $\overrightarrow{b} = (b_1,\dotsc, b_{d-1})$.

\subsubsection*{Qudit ZX-part of the rules}

\begin{gather}
  \input{axioms/gengspiderfusedit.tikz}
  \tag{S1}\label{rule:S1}\refstepcounter{equation}
\end{gather}
\begin{multicols}{2}
  \noindent
  \begin{gather}
    \input{axioms/s2qudit.tikz}
    \tag{S2}\label{rule:S2}\refstepcounter{equation} \\
    \input{axioms/dcomwtha0.tikz}
    \tag{D1}\label{rule:D1}\refstepcounter{equation} \\
    \input{axioms/rdotaemptydit0.tikz}
    \tag{Ept}\label{rule:Ept}\refstepcounter{equation} \\
    \input{axioms/b2qudit.tikz}
    \tag{B2}\label{rule:B2}\refstepcounter{equation} \\
    \input{axioms/k1copy.tikz}
    \tag{K0}\label{rule:K0}\refstepcounter{equation} \\ \columnbreak
    \input{axioms/pimultiplecpdit.tikz}
    \tag{K1}\label{rule:K1}\refstepcounter{equation} \\
    \input{axioms/k2adit.tikz}
    \tag{K2}\label{rule:K2}\refstepcounter{equation} \\
    \input{axioms/zerotoreddit0.tikz}
    \tag{Zer}\label{rule:Zer}\refstepcounter{equation} \\
    \input{axioms/p1sdit2.tikz}
    \tag{P1}\label{rule:P1}\refstepcounter{equation} \\
    \input{axioms/hhdaggerchangedit.tikz}
    \tag{H1}\label{rule:H1}\refstepcounter{equation}
  \end{gather}
\end{multicols}
\vspace{-0.5cm}
\begin{equation*}
  \text{where}\quad
  \protect\overleftarrow{a}=(a_{d-1}, \dotsc, a_1),\quad
  \protect\overrightarrow{ab}=(a_1 b_1,\dotsc, a_{d-1} b_{d-1}),
  \quad\text{and}\quad
  \protect{k_j(\overrightarrow{a})}=\left(\frac{a_{1-j}}{a_{d-j}}, \dotsc, \frac{a_{d-1-j}}{a_{d-j}}\right)
\end{equation*}

\subsubsection*{Discussion of the qudit ZX-part of the rules}
This part of the ruleset is the direct generalisation of the rules of the qubit ZX-calculus.
One of the most differentiating features of our calculus is that it is not flexsymmetric~\cite{caretteWhenOnlyTopology2021}, which means that we cannot see the ZXW diagrams as open graphs.
Another interesting difference is that there are two ways to change the colour of spiders based on where the $H$ and $H^\dagger$ boxes are located.
Further to the axioms above, we have some additional notations or lemmas that are familiar equalities in the qudit ZX-calculus: \eqref{rule:S3}, \eqref{rule:S4}, \eqref{rule:HDagger}, \eqref{rule:HZ}, \eqref{rule:HX}, \eqref{rule:Du}, \eqref{rule:Mu}, \eqref{rule:Hopf}

\subsubsection*{Qudit ZW-part of the rules}

\begin{multicols}{2}
  \noindent
  \begin{gather}
    \input{axioms/phasecopydit.tikz}
    \tag{Pcy}\label{rule:Pcy}\refstepcounter{equation} \\
    \input{axioms/wsymetrydit.tikz}
    \tag{Sym}\label{rule:Sym}\refstepcounter{equation} \\
    \input{axioms/w-bialgebra.tikz}
    \tag{BZW}\label{rule:BZW}\refstepcounter{equation} \\ \columnbreak
    \input{axioms/additiondit.tikz}
    \tag{AD}\label{rule:AD}\refstepcounter{equation} \\
    \input{axioms/associatedit.tikz}
    \tag{Aso}\label{rule:Aso}\refstepcounter{equation} \\
    \input{axioms/w-w-algebra.tikz}
    \tag{WW}\label{rule:WW}\refstepcounter{equation}
  \end{gather}
\end{multicols}

\subsubsection*{Discussion of the qudit ZW-part of the rules}
The above part of the ruleset is also a direct generalisation of the rules of the qubit ZW-calculus~\cite{hadzihasanovicDiagrammaticAxiomatisationQubit2015} except for the rule~\eqref{rule:WW}.
In the qubit ZW-calculus, there is a bialgebra rule for the W node and its transpose that involves a fermion swap, but this rule does not hold for the higher dimensional W node.
Nevertheless, we find the \eqref{rule:WW} rule that allows for a bialgebra between the W node and its transpose without a fermionic swap, given we restrict inputs and outputs via a projection with a W node.

\subsubsection*{Qudit ZXW-part of the rules}

\begin{multicols}{2}
  \noindent
  \begin{gather}
    \input{axioms/triangleocopydit.tikz}
    \tag{Bs0}\label{rule:Bs0}\refstepcounter{equation} \\
    \input{axioms/trianglepicopydit2.tikz}
    \tag{Bsj}\label{rule:Bsj}\refstepcounter{equation} \\
    \input{axioms/trialgebra.tikz}
    \tag{TA}\label{rule:TA}\refstepcounter{equation} \\
    \input{axioms/hadamard-decomposition2.tikz}
    \tag{HD}\label{rule:HD}\refstepcounter{equation} \\\columnbreak
    \input{axioms/z-push-v.tikz}
    \tag{ZV}\label{rule:ZV}\refstepcounter{equation} \\
    \input{axioms/zbox-v-decomposition-2.tikz}
    \tag{VA}\label{rule:VA}\refstepcounter{equation} \\
    \input{axioms/v-push-w.tikz}
    \tag{VW}\label{rule:VW}\refstepcounter{equation} \\
    \input{axioms/k1zstate.tikz}
    \tag{KZ}\label{rule:KZ}\refstepcounter{equation}
  \end{gather}
\end{multicols}
\vspace{-0.5cm}
\begin{align*}
  \text{where}&\quad
  T_j=\overbrace{(\underbrace{0,\dotsc, 1}_{d-j}, \dotsc, 0)}^{d-1},\quad
  e_1, \dotsc, e_n \in \{ 1, \dotsc, d-1\},\quad
  \overrightarrow{a_{d-1}} = \left(a_{d-1}, a_{d-1}, \dotsc, a_{d-1}\right)
\end{align*}

\subsubsection*{Discussion of the qudit ZXW-part of the rules}
This part of the ruleset describes the newly discovered interaction of the Z, X and W algebras that have not been presented before.
These rules capture complex relations that would otherwise be cumbersome to express if limited to solely ZX- or solely ZW-calculus.
Firstly, the \eqref{rule:Bs0} and the \eqref{rule:Bsj} rules show how the W node interacts with computational basis states.

One of the most important rules is the trialgebra rule \eqref{rule:TA} that allows us to switch between the W node and the X spider.
We can extend this axiom to the general case with any number of W nodes and X spiders as long as each Z spider is also connected to a common W node (see \cref{gerneraltrialgebralm}).

To explain the rule \eqref{rule:KZ}, consider the action of the right hand side on the basis states $\ket{b_i}$ where the wires are numbered $0$ to $n$.
Then the projection of the W node enforces that exactly one of the $b_i$ is $\ket{-1}$, and the rest are $\ket 0$.
Meanwhile, the other projection of the red node enforces that:
\[
  \sum_{i = 1}^{n} e_i b_i \equiv 0 \Mod{d}
\]
Therefore, the one $b_i$ that is $\ket{-1}$ must be the leftmost wire.

The rule \eqref{rule:HD} is the decomposition of the Hadamard box.
In this decomposition, each Z-spider with phase $K_r$ corresponds to the $r$-th row of the qudit Hadamard gate.
The rest of the diagram combines these rows.

We have three rules directly related to the $V$ box.
Firstly, the \eqref{rule:VA} rule enables us to combine $d - 1$ Z boxes with single non-zero phases $(0, \cdots, 0, a_i)$ into one Z box with a phase vector constitutes those $d - 1$ phases.
The \eqref{rule:VW} and the \eqref{rule:ZV} rules describe how the $V$ box interacts with the $W$ node and Z box states.

\endgroup 

%% file: figures/axioms/zerotoreddit0.tikz
\begin{tikzpicture}
	\begin{pgfonlayer}{nodelayer}
		\node [style=gbox] (0) at (-1, 0.5) {$0$};
		\node [style=rn] (1) at (1, 0.5) {};
		\node [style=none] (2) at (-1, -0.5) {};
		\node [style=none] (3) at (1, -0.5) {};
		\node [style=none] (4) at (0, 0) {=};
	\end{pgfonlayer}
	\begin{pgfonlayer}{edgelayer}
		\draw (2.center) to (0);
		\draw (1) to (3.center);
	\end{pgfonlayer}
\end{tikzpicture}

%% file: figures/definitions/s-scalar.tikz
\begin{tikzpicture}
	\begin{pgfonlayer}{nodelayer}
		\node [style={rn_phase}] (0) at (0, 0.75) {$K_1$};
		\node [style=gbox] (1) at (0, -0.5) {$s$};
	\end{pgfonlayer}
	\begin{pgfonlayer}{edgelayer}
		\draw (0) to (1);
	\end{pgfonlayer}
\end{tikzpicture}

%% file: proofidea.tex
\subsection{Strategy of the completeness proof}

\newcommand{\proofstrategyscale}{0.8}

In this section, we describe the idea of the completeness proof of qudit ZXW-calculus.
The technique we use, i.e.\@ leveraging the matrix itself as the unique normal form, has been used before in the context of proving the completeness of qubit algebraic ZX-calculus~\cite{wangAlgebraicCompleteAxiomatisation2022} and qubit ZW-calculus~\cite{hadzihasanovicDiagrammaticAxiomatisationQubit2015}.

By completeness of qudit ZXW-calculus, we mean that for any two diagrams  $D_1$ and $D_2$, if $\interp{D_1} = \interp{D_2}$, then we can derive that $D_1 = D_2$ from the rules of qudit ZXW-calculus.
Due to the following map-state duality:
\begin{equation*}
  \scalebox{\proofstrategyscale}{
    \vspace{-2cm}
    \input{proof-idea/mapstatedual3.tikz}
  }
  \label{eq:maptostate}
\end{equation*}
we only need to consider state diagrams (diagrams without any input).

Now assume that we have two qudit ZXW state diagrams $D_1$ and $D_2$ such that $\interp{D_1} = \interp{D_2}$.
We need to show that $D_1= D_2$.
Since $\interp{D_1} = \interp{D_2}$,  $D_1$ and $D_2$ must have the same normal form $D$.
Therefore, if we can rewrite both $D_1$ and $D_2$ into $D$, then by inverting the rewriting (which is a series of equalities) from $D_2$ to $D$ we obtain that $D_1 = D_2$.
Thus, the proof strategy for completeness is to show that any state diagram can be rewritten into a normal form.

To rewrite an arbitrary state diagram into a normal form, we need to analyse its structure.
We first note that each state diagram $D_s$ has the following form:
\begin{equation*}
  \scalebox{\proofstrategyscale}{
    \input{proof-idea/generalstatedm.tikz}
  }
  \label{eq:generalstatediam}
\end{equation*}
where $A_1,\, A_2,\, \cdots,\, A_n$ are diagrams that are the parallel compositions of generators given in Section~\ref{subsec:generators_n_interpretation}.
\begin{remark}
  We treat scalar diagrams that are composed of multiple generators as special state diagrams.
  The only generator that can be a scalar on its own is the \input{proof-idea/scalar-generator.tikz} diagram.
  To make the following proof steps are applicable, we rewrite such diagrams into \input{proof-idea/scalar-generator-transform.tikz} using the rule~\eqref{rule:S1}.
\end{remark}
As a consequence, we can rewrite $D_s$ into its normal form according to the following steps:
\begin{enumerate}
  \item We rewrite $A_1$, which is the tensor product of generators without inputs, into its normal form $N_1$:
  \[
    \scalebox{\proofstrategyscale}{
      \input{proof-idea/a1_to_n1.tikz}
    }
  \]
  \item We bend the top of the diagram so that the generators in $A_2$ become top-bended state diagrams:
  \[
    \scalebox{\proofstrategyscale}{
      \input{proof-idea/sequentialbend.tikz}
    }
  \]
  \item We convert the top-bent $A_2$ diagram into its corresponding normal form $N_2$:
  \[
    \scalebox{\proofstrategyscale}{
      \vspace{-2cm}
      \input{proof-idea/a2_to_n2.tikz}
    }
  \]
  \item We rewrite the tensor product of the two normal forms into a single normal form denoted with $N_{1,2}$:
  \[
    \scalebox{\proofstrategyscale}{
      \input{proof-idea/n1_otimes_n2.tikz}
    }
  \]
  \item We reduce the partial traces of the normal form (i.e.\@ connections of two outputs of a normal form with a cup) into another normal form $N_{1,2}^\prime$:
  \[
    \scalebox{\proofstrategyscale}{
      \input{proof-idea/n12_partial_trace.tikz}
    }
  \]
  \item We repeat from step 2 for the rest of the diagram, that is from $A_3$ to $A_n$.
\end{enumerate}
If we follow the above steps, we acquire the normal form of $D_s$.

To summarise the requirements, as marked with the dashed lines in the above equations, the following statements are necessary to prove the completeness of the ZXW-calculus:
\begin{itemize}
  \item All generators bent in state diagrams or already being state diagrams can be rewritten into their normal forms.
  \item The tensor product of any two normal forms can be rewritten into a single normal form.
  \item A partial-traced normal form can be rewritten into a normal form.
\end{itemize}

%% file: figures/proof-idea/scalar-generator.tikz
\begin{tikzpicture}
	\begin{pgfonlayer}{nodelayer}
		\node [style=gbox] (0) at (0, 0) {$\overrightarrow x$};
	\end{pgfonlayer}
\end{tikzpicture}

%% file: figures/proof-idea/scalar-generator-transform.tikz
\begin{tikzpicture}
	\begin{pgfonlayer}{nodelayer}
		\node [style=gbox] (2) at (0, 0.5) {$\overrightarrow x$};
		\node [style=gn] (3) at (0, -0.5) {};
	\end{pgfonlayer}
	\begin{pgfonlayer}{edgelayer}
		\draw (3) to (2);
	\end{pgfonlayer}
\end{tikzpicture}

%% file: completenessnoproof.tex
\subsection{The completeness proof}\label{subsec:completeness-no-proof}

In this section, we show the completeness of the ZXW-calculus.
\begin{restatable}[Completeness]{theorem}{zxwcompleteness}
  ZXW-calculus is complete for arbitrary finite dimensions.
\end{restatable}
\begin{proof}
  To prove the above theorem, we first need to show that the Z box, W node, and Hadamard box can be rewritten into a normal form.
  We do not need to show this for the other generators, the swap and the identity, as their normal form can be derived from the above generators by applying the steps in the previous section.
  Then, we need to show that any partial trace of a normal form can be transformed into another normal form.
  Finally, we need to rewrite the tensor product of two normal forms into a normal form.
  The above statements are formalised in \cref{lem:zboxnf,lem:wnodenf,lem:hadamardnf,lem:nfpartialtrace,lem:nftensorproduct};
  furthermore, we present the proofs for these lemmas in the appendix.
\end{proof}
\begin{restatable}[Z box]{lemma}{zboxnf}
  \label{lem:zboxnf}
  \[
    \input{qudit-nf-generators/zboxshort.tikz}
  \]
\end{restatable}
\begin{restatable}[W node]{lemma}{wnodenf}
  \label{lem:wnodenf}
  \[
    \input{qudit-nf-generators/wshort.tikz}
  \]
\end{restatable}

\begin{restatable}[Hadamard box]{lemma}{hadamardnf}
  \label{lem:hadamardnf}
  \[
    \input{qudit-nf-generators/hadamardshort.tikz}
  \]
\end{restatable}
\begin{restatable}[Partial trace]{lemma}{nfpartialtrace}
  \label{lem:nfpartialtrace}
  \[
    \input{qudit-nf-trace/pfshort.tikz}
  \]
  where $\Sigma_k$ corresponds to the elements of the partial trace over the $s$ and $t$ indices.
  That is, $\Sigma_k$ is the sum of the boxes $a_{k_0}, \cdots ,a_{k_{d \minu 1}}$ with such multiplier connections that satisfy $e_{x,k_0} = \cdots = e_{x,k_{d \minu 1}}$ for all $x \in \{0,\, \ldots,\, d - 1\} \setminus \{s, t\}$ and $e_{s,k_{y}} = e_{t,k_y} = y$ for $0 \leq y \leq d - 1$.
\end{restatable}
\begin{restatable}[Tensor product]{lemma}{nftensorproduct}
  \label{lem:nftensorproduct}
  \[
    \input{qudit-tensor/tensorshort.tikz}
  \]
  where $M=d^m-1,\, N=d^n-1$.
\end{restatable}

%% file: conclusion.tex
In this paper, we introduce and prove the completeness of the qudit ZXW-calculus, a significant step in showing that the ZXW-calculus is a useful graphical calculus for quantum computing that unifies the ZX and ZW calculi.
The addition of the W node to the ZX-calculus allows us to do (controlled) summation, differentiation, and integration of ZX diagrams.
On the other hand, the ZXW trialgebra rule allows switching between X and W spiders, which makes rewriting and reasoning easier in many scenarios.

Our completeness proof utilised the matrix itself as a unique normal form for any ZXW diagram.
We not only prove that there always exists a series of rules to rewrite any ZXW diagram to its unique normal form, but moreover present a deterministic procedure to do so.
Immediately following from our result is the guarantee that the present ruleset suffices to prove the equivalence of any two mathematically equal ZXW diagrams in any qudit dimension, i.e.\@ it has no \enquote{missing rules}.

Note that in this work thus far, the question of ruleset minimality was only pursued lightly.
Through challenging the necessity of each rule and attempting to derive it from the other rules, it is anticipated, as has been achieved for the qubit ZX-calculus~\cite{vilmartNearMinimalAxiomatisationZXCalculus2019}, that the complete ruleset may be further streamlined and henceforth capture the key interactions more succinctly and conveniently.

As a generalisation of the ZX and ZW calculi, the ZXW-calculus should be able to tackle the problems accessible by both calculi, including circuit compilation, circuit simulation, reasoning with entanglement, and MBQC\@.
The ZXW-calculus has also branched out in new directions, such as quantum machine learning and quantum chemistry.
As future work, the ZXW-calculus can be applied to the study of quantum optics, where the W node plays an important role.
Another natural direction is to extend the ZXW calculus to handle problems such as speeding up tensor network contraction algorithms and solving counting problems~\cite{debeaudrapTensorNetworkRewriting2021}.
We leave this for future work.
In particular, to make progress on circuit compilation and MBQC, the development of techniques for circuit extraction and new flow conditions is likely necessary.

A flexsymmetric~\cite{caretteWhenOnlyTopology2021} version of the qudit ZXW-calculus can be obtained by defining the X-spider and the W node differently:
\begin{itemize}
  \item the flexsymmetric X-spider can be defined by replacing the $H^\dagger$-boxes with $H$-boxes in~\eqref{rule:HZ},
  \item the flexsymmetric W node can be defined using the diagram that corresponds to the 3-qudit standardised W state~\eqref{eq:wstate}.
\end{itemize}
An interesting path for future work would be to study the rules of this flexsymmetric qudit ZXW-calculus.

Our explicit representation of finite matrices in the ZXW-calculus makes clear the connection between quantum circuits and the matrices they implement, readily becoming a tool for tackling a number of open problems in qudit circuit synthesis.
For future work, it should also be possible to extend this representation to the qufinite case~\cite{wangQufiniteZXcalculusUnified2022}, where tensors can have unequal dimensions on each axis, thus allowing us to represent any finite tensor in the ZXW-calculus.
In the future, we can better reason about qudit quantum circuits by applying our result to establish completeness for the qudit ZX-calculus (without W as a generator), building upon the recent translation of any prime-dimensional qudit W state to a circuit directly representable in ZX~\cite{yehScalingStateCircuits2023}.

It is important to remember that completeness proofs do not give us a constructive algorithm to \emph{efficiently} rewrite one diagram into another, and it is necessary to develop further techniques for ZXW to achieve this.
The efficient rewriting of ZXW diagrams is a topic of ongoing research.

%% file: lemmas.tex
\subsection{Lemmas and Proofs}\label{sec:appendix}

\begin{lemma}
  \label{kjgalm}\cite{wangQufiniteZXcalculusUnified2022}
  \[
    \input{lemmas/kjga.tikz}
  \]
  where $\overrightarrow{a}=(a_1,\cdots, a_{d-1})$, $j \in \{ 1,\cdots, d-1\}.$
\end{lemma}
\begin{proof}
  \[
    \input{lemmas/kjgaprf.tikz}
  \]
\end{proof}

\begin{lemma}
  \label{scalargeneralmultlm}\cite{wangQufiniteZXcalculusUnified2022}
  \input{lemmas/scalargeneralmult.tikz}
\end{lemma}
\begin{proof}
  \[
    \input{lemmas/scalargeneralmultprf.tikz}
  \]
\end{proof}

\begin{lemma}
  \label{zeroemptyditlm}\cite{wangQufiniteZXcalculusUnified2022}
  \input{lemmas/zeroemptydit.tikz}
\end{lemma}
\begin{proof}
  \[
    \input{lemmas/zeroemptyditprf.tikz}
  \]
\end{proof}

\begin{lemma}
  \label{scalarinverseditlm}
  Suppose $a \neq 0$.
  Then \input{lemmas/scalarinversedit.tikz}
\end{lemma}
\begin{proof}
  This follows directly from \cref{scalargeneralmultlm,zeroemptyditlm}.
\end{proof}

\begin{lemma}
  \label{hilm}\cite{wangQufiniteZXcalculusUnified2022}
  \input{lemmas/h1scalar.tikz}
\end{lemma}
\begin{proof}
  \[
    \input{lemmas/h1scalarprf.tikz}
  \]
\end{proof}

\begin{lemma}
  \label{hheqhdhd}
  \input{lemmas/hheqhdhd.tikz}
\end{lemma}
\begin{proof}
  \[
    \input{lemmas/hheqhdhdpf.tikz}
  \]
\end{proof}

\begin{lemma}
  \label{s4lm}\cite{wangQufiniteZXcalculusUnified2022}
  \begin{gather}
    \input{lemmas/redspider0pfusedit2.tikz}
    \tag{S4}\label{rule:S4}
  \end{gather}
\end{lemma}
\begin{proof}
  \[
    \input{lemmas/redspider0pfusedit2prf.tikz}
  \]
  where $u_{m,n+1}=d^{\frac{m+n-1}{2}}-1$, $u_{t+1,s}=d^{\frac{s+t-1}{2}}-1$, and $u_{m+t,n+s}=d^{\frac{m+n+s+t-2}{2}}-1$.
  The other equalities can be proved similarly using the same rules.
\end{proof}

\begin{lemma}
  \label{redspiderforgrlm}\cite{wangQufiniteZXcalculusUnified2022}
  \begin{gather}
    \input{lemmas/redspiderforgr.tikz}
    \tag{HX}\label{rule:HX}
  \end{gather}
  where $v_{m,n}= d^{\frac{-m-n+2}{2}}-1$, $j \in \{ 0,1,\cdots, d-1\}$.
  We also call this equality~\eqref{rule:HX}.
\end{lemma}
\begin{proof}
  \begin{align*}
    &\input{lemmas/redspiderforgrprf-1.tikz} \\
    &\input{lemmas/redspiderforgrprf-2.tikz}
  \end{align*}
  where $u_{m,n}= d^{\frac{m+n-2}{2}}-1$, $v_{m,n}= d^{\frac{-m-n+2}{2}}-1$, $j \in \{ 0,1,\cdots, d-1\}$.
\end{proof}

\begin{lemma}
  \label{b3quditlm}
  \begin{gather}
    \input{lemmas/b3qudit.tikz}
    \tag{B3}\label{rule:B3}
  \end{gather}
\end{lemma}
\begin{proof}
  \[
    \input{lemmas/b3quditprf.tikz}
  \]
  where $u_{2,1}= d^{\frac{1}{2}}-1$, $v_{1,0}= d^{\frac{1}{2}}-1$.
\end{proof}


\begin{lemma}
  \label{dboxsquarelm}\cite{wangQufiniteZXcalculusUnified2022}
  \input{lemmas/dboxsquare.tikz}
\end{lemma}
\begin{proof}
  \[
    \input{lemmas/dboxsquareprf.tikz}
  \]
\end{proof}

\begin{lemma}
  \label{dboxgdotlm}\cite{wangQufiniteZXcalculusUnified2022}
  \[
    \input{lemmas/dboxgdot.tikz}
  \]
\end{lemma}
\begin{proof}
  \[
    \input{lemmas/dboxgdotprf.tikz}
  \]
  The second equality can be proved similarly.
\end{proof}

\begin{lemma}
  \label{dualiserslm}\cite{wangQufiniteZXcalculusUnified2022}
  \input{lemmas/dualisers.tikz}
\end{lemma}
\begin{proof}
  \begin{align*}
    &\input{lemmas/dualisersprf-1.tikz} \\
    &\input{lemmas/dualisersprf-2.tikz} \\
    &\input{lemmas/dualisersprf-3.tikz}
  \end{align*}
\end{proof}

\begin{lemma}
  \label{hadslidegnlm}\cite{wangQufiniteZXcalculusUnified2022}
  \begin{align*}
    &\scalebox{.9}{
      \input{lemmas/hadslidegn-1.tikz}
    } \\
    &\scalebox{.9}{
      \input{lemmas/hadslidegn-2.tikz}
    }
  \end{align*}
\end{lemma}
\begin{proof}
  \[
    \input{lemmas/hadslidegnprf.tikz}
  \]
  The other equalities can be proved similarly.
\end{proof}

\begin{lemma}
  \label{dboxslidegrnlm}\cite{wangQufiniteZXcalculusUnified2022}
  \begin{align*}
    &\input{lemmas/dboxslidegrn-1.tikz}\\
    &\input{lemmas/dboxslidegrn-2.tikz}
  \end{align*}
\end{lemma}
\begin{proof}
  \[
    \input{lemmas/dboxslidegrnprf.tikz}
  \]
  The other equalities can be proved similarly.
\end{proof}

\begin{lemma}
  \label{dboxbell}
  \begin{align*}
    &\input{lemmas/dboxbell-1.tikz} \\
    &\input{lemmas/dboxbell-2.tikz}
  \end{align*}
\end{lemma}
This follows from \cref{dboxslidegrnlm,dboxsquarelm}.

\begin{lemma}
  \label{swaprgcapcupditlm}\cite{wangQufiniteZXcalculusUnified2022}
  \[
    \input{lemmas/swaprgcapcupdit.tikz}
  \]
\end{lemma}
\begin{proof}
  We only prove the first equality, the others can be proved similarly.
  \[
    \input{lemmas/swaprgcapcupditprf.tikz}
  \]
\end{proof}

\begin{lemma}
  \label{grconnectkslm}\cite{wangQufiniteZXcalculusUnified2022}
  \[
    \input{lemmas/grconnectks.tikz}
  \]
\end{lemma}
\begin{proof}
  \[
    \input{lemmas/grconnectksprf.tikz}
  \]
\end{proof}

\begin{lemma}
  \label{controlnotslideditlm}\cite{wangQufiniteZXcalculusUnified2022}
  \begin{align*}
    &\input{lemmas/controlnotslidedit-1.tikz} \\
    &\input{lemmas/controlnotslidedit-2.tikz}
  \end{align*}
\end{lemma}
\begin{proof}
  We only prove the following equality, the others can be proved similarly.
  \[
    \input{lemmas/controlnotslideditprf.tikz}
  \]
\end{proof}

\begin{lemma}
  \label{xdualiserlm}
  \begin{align*}
    &\input{lemmas/xdualiser-1.tikz}\\
    &\input{lemmas/xdualiser-2.tikz}
  \end{align*}
\end{lemma}
\begin{proof}
  \begin{align*}
    &\input{lemmas/xdualiserpf-1.tikz}\\
    &\input{lemmas/xdualiserpf-2.tikz}
  \end{align*}
  The other equalities can be proved similarly.
\end{proof}

\begin{lemma}
  \label{slidecupditlm}\cite{wangQufiniteZXcalculusUnified2022}
  \begin{align*}
    &\input{lemmas/slidecupdit-1.tikz}\\
    &\input{lemmas/slidecupdit-2.tikz}
  \end{align*}
  where $\overleftarrow{a}=(a_{d-1}, \cdots, a_1)$.
\end{lemma}
\begin{proof}
  \begin{align*}
    &\input{lemmas/slidecupditprf-1.tikz}\\
    &\input{lemmas/slidecupditprf-2.tikz}
  \end{align*}
  The other equalities can be proved similarly.
\end{proof}

\begin{lemma}
  \label{copyvarsditlm}\cite{wangQufiniteZXcalculusUnified2022}
  \[
    \input{lemmas/copyvarsdit.tikz}
  \]
\end{lemma}
\begin{proof}
  \begin{align*}
    &\input{lemmas/copyvarsditprf-1.tikz}\\
    &\input{lemmas/copyvarsditprf-2.tikz}
  \end{align*}
  The other equalities can be proved similarly.
\end{proof}

\begin{lemma}
  \label{spider0tordotslm}\cite{wangQufiniteZXcalculusUnified2022}
  \[
    \input{lemmas/spider0tordots.tikz}
  \]
\end{lemma}
\begin{proof}
  \[
    \input{lemmas/spider0tordotsprf.tikz}
  \]
\end{proof}

\begin{lemma}
  \label{hopfditlm}\cite{wangQufiniteZXcalculusUnified2022}
  \begin{gather}
    \input{lemmas/hopfdit.tikz}
    \tag{Hopf}\label{rule:Hopf}
  \end{gather}
\end{lemma}
\begin{proof}
  \begin{align*}
    &\input{lemmas/hopfditprf-1.tikz}\\
    &\input{lemmas/hopfditprf-2.tikz}
  \end{align*}
\end{proof}

\begin{lemma}
  \label{redgreenmchangelm}\cite{wangQufiniteZXcalculusUnified2022}
  \[
    \input{lemmas/redgreenmchange.tikz}
  \]
  where $1 \leq k \leq d-1$.
\end{lemma}
\begin{proof}
  \[
    \input{lemmas/redgreenmchangeprf.tikz}
  \]
  The second equality can be proved similarly.
\end{proof}

\begin{lemma}
  \label{multiplierpushlm}
  Suppose $x \in \{0, \dotsc, d - 1 \}$.
  Then
  \[
    \input{lemmas/multipliers/multiplier-push.tikz} \qquad
    \qquad
    \input{lemmas/multipliers/multiplier-push-dual.tikz}
  \]
\end{lemma}
\begin{proof}
  \[
    \input{lemmas/multipliers/multiplier-push-proof.tikz}
  \]
  The second equality can be proved similarly.
\end{proof}

\begin{lemma}
  \label{multipliermultlm}
  Suppose $x, y \in \{0, \dotsc, d - 1 \}$.
  Then
  \[
    \input{lemmas/multipliers/multiplier-mult.tikz}
  \]
\end{lemma}
\begin{proof}
  \[
    \input{lemmas/multipliers/multipliermultprf.tikz}
  \]
\end{proof}

\begin{lemma}
  \label{multiplieraddlm}
  \[
    \input{lemmas/multipliers/multiplier-add.tikz}
  \]
\end{lemma}
This follows directly from the definition of a multiplier.

\begin{lemma}
  \label{dualisermultipliermultlm}
  Suppose $x \in \{0, \dotsc, d - 1 \}$.
  Then
  \[
    \input{lemmas/multipliers/dualiser-multiplier-mult.tikz}
  \]
\end{lemma}
This follows directly from rule~\eqref{rule:P1} and \cref{multipliermultlm}.

\begin{lemma}
  \label{dboxgcopylm}\cite{wangQufiniteZXcalculusUnified2022}
  \[
    \input{lemmas/dboxgcopy.tikz}
  \]
\end{lemma}
This follows from rule~\eqref{rule:P1} and \cref{multiplierpushlm,grconnectkslm}.

\begin{lemma}
  \label{xtransposelm}
  \[
    \input{lemmas/xtranspose.tikz}
  \]
\end{lemma}
\begin{proof}
  \[
    \input{lemmas/xtransposepf.tikz}
  \]
  The second equality can be proved similarly.
\end{proof}

\begin{lemma}
  \label{wdecomplm}
  \[
    \input{lemmas/wdecomp.tikz}
  \]
\end{lemma}
\begin{proof}
  \begin{align*}
    \input{lemmas/wdecompprf-0.tikz} \quad
    &\input{lemmas/wdecompprf-1.tikz} \\
    &\input{lemmas/wdecompprf-2.tikz}
  \end{align*}
\end{proof}

\begin{lemma}
  \label{wunitlm}
  \[
    \input{lemmas/wunit.tikz}
  \]
\end{lemma}
\begin{proof}
  \[
    \input{lemmas/wunitprf.tikz}
  \]
\end{proof}

\begin{lemma}
  \label{trialgbzwlm}
  Suppose $m\geq 2$.
  Then
  \[
    \input{lemmas/trialgbzw.tikz}
  \]
\end{lemma}
\begin{proof}
  We prove this by induction.
  For $m=2$, it just follows from the rules~\eqref{rule:TA} and~\eqref{rule:BZW}.
  Assume that it holds for $m = k$.
  Then
  \[
    \input{lemmas/trialgbzwprf.tikz}
  \]
\end{proof}

\begin{lemma}
  \label{trialgebraallconnectionlm}
  \[
    \input{lemmas/trialgebraallconnection.tikz}
  \]
\end{lemma}
This follows from the rule~\eqref{rule:BZW} and \cref{trialgbzwlm}.

\begin{lemma}
  \label{gerneraltrialgebralm}
  \[
    \input{lemmas/gerneraltrialgebra.tikz}
  \]
\end{lemma}
This equality can be derived directly from \cref{trialgebraallconnectionlm} by splitting the W spider using the rule~\eqref{rule:WN}.

\begin{lemma}
  \label{2wadditionlm}
  \[
    \input{lemmas/2waddition.tikz}
  \]
\end{lemma}
\begin{proof}
  \begin{align*}
    \input{lemmas/2wadditionprf-0.tikz} \quad
    &\input{lemmas/2wadditionprf-1.tikz} \\
    &\input{lemmas/2wadditionprf-2.tikz} \\
  \end{align*}
\end{proof}

\begin{lemma}
  \label{K1Wgbox1lm}
  \[
    \input{lemmas/K1Wgbox1.tikz}
  \]
\end{lemma}
\begin{proof}
  \[
    \input{lemmas/K1Wgbox1prf.tikz}
  \]
\end{proof}

\begin{lemma}
  \label{wtransposelm}
  \[
    \input{lemmas/wtranspose.tikz}
  \]
\end{lemma}
This follows directly form rules~\eqref{rule:WN},~\eqref{rule:S3}, and~\eqref{rule:S1}.

\begin{lemma}
  For $\overrightarrow{a} = (a_1, a_2, \, \ldots \, , a_{d - 1})$, we have
  \label{zbox-single-decompositionlm}
  \[
    \input{lemmas/zbox-single-decomposition.tikz}
  \]
\end{lemma}
\begin{proof}
  \[
    \input{lemmas/zbox-single-decompositionprf.tikz}
  \]
\end{proof}

\begin{lemma}
  \label{zbox-double-decompositionlm}
  For $\overrightarrow{a} = (a_1, a_2, \, \ldots \, , a_{d - 1})$ and $\overrightarrow{b} = (b_1, b_2, \, \ldots \, , b_{d - 1})$, we have
  \[
    \input{lemmas/zbox-double-decomposition.tikz}
  \]
\end{lemma}
\begin{proof}
  \[
    \input{lemmas/zbox-double-decompositionprf.tikz}
  \]
\end{proof}

%% file: figures/lemmas/scalargeneralmult.tikz
\begin{tikzpicture}
	\begin{pgfonlayer}{nodelayer}
		\node [style=gbox] (0) at (-3.75, 0) {$a \minu 1$};
		\node [style=gbox] (2) at (-1.75, 0) {$b \minu 1$};
		\node [style=gbox] (4) at (2, 0) {$ab \minu 1$};
		\node [style=none] (6) at (0, 0) {$=$};
	\end{pgfonlayer}
\end{tikzpicture}

%% file: appendix.tex
\nftounf*
\begin{proof}
  First, we show that we can change the order of any two neighbouring Z boxes, $a_k$ and $a_{k + 1}$:
  \begin{align*}
    &\input{lemmas/nf-swap-1.tikz}\\
    &\input{lemmas/nf-swap-2.tikz}
  \end{align*}
  Using such neighbouring swaps, we can sort the Z boxes based on the lexicographical ordering of their multiplier connection.
  This can be done using an applicable algorithm such as bubble sort, after which our normal form is in its unique form.
\end{proof}

\subsection{Generators}
\input{qudit-generators.tex}

\subsection{Partial trace}
\input{qudit-trace.tex}

\subsection{Tensor product}
\input{qudit-tensor.tex}

%% file: qudit-generators.tex
\zboxnf*
\begin{proof}
  \begin{align*}
    \input{qudit-nf-generators/zbox-0.tikz} \quad
    &\input{qudit-nf-generators/zbox-1.tikz} \\
    &\input{qudit-nf-generators/zbox-2.tikz}
  \end{align*}
\end{proof}

\wnodenf*
\begin{proof}
  \begin{align*}
    \input{qudit-nf-generators/w-0.tikz} \quad
    &\input{qudit-nf-generators/w-1.tikz} \\
    &\input{qudit-nf-generators/w-2.tikz} \\
    &\input{qudit-nf-generators/w-3.tikz} \\
    &\input{qudit-nf-generators/w-4.tikz}
  \end{align*}
\end{proof}

\begin{lemma}
  \label{vector-combine-with-multiplierslm}
  \[
    \input{lemmas/vector-combine-with-multipliersshort.tikz}
  \]
\end{lemma}
\begin{proof}
  \begin{align*}
    \input{lemmas/vector-combine-with-multipliers-0.tikz} \quad
    &\input{lemmas/vector-combine-with-multipliers-1.tikz} \\
    &\input{lemmas/vector-combine-with-multipliers-2.tikz} \\
    &\input{lemmas/vector-combine-with-multipliers-3.tikz}
  \end{align*}
\end{proof}

\hadamardnf*
\begin{proof}
  \begin{align*}
    \input{qudit-nf-generators/hadamard-0.tikz} \quad
    &\input{qudit-nf-generators/hadamard-1.tikz} \\
    &\input{qudit-nf-generators/hadamard-2.tikz} \\
    &\input{qudit-nf-generators/hadamard-3.tikz} \\
    &\input{qudit-nf-generators/hadamard-4.tikz} \\
    &\input{qudit-nf-generators/hadamard-5.tikz} \\
    &\input{qudit-nf-generators/hadamard-6.tikz}
  \end{align*}
\end{proof}

%% file: qudit-trace.tex
\begin{lemma}
  \label{2pinkfusionlm}
  \[
    \input{qudit-nf-trace/lem1.tikz}
  \]
\end{lemma}
\begin{proof}
  \[
    \input{qudit-nf-trace/lem1pf.tikz}
  \]
\end{proof}

\nfpartialtrace*
\begin{proof}
  First of all,
  \begin{align*}
    &\input{qudit-nf-trace/pf1-1.tikz} \\
    &\input{qudit-nf-trace/pf1-2.tikz}
  \end{align*}
  Now, the green boxes $\{a_i \}_{0\leq i \leq d^m-1}$ can be divided into two classes based on their number of connections to the $s \wedge t$ node:
  \begin{enumerate}
    \item $e_{s,i} \equiv e_{t,i} \Mod{d}$.
    \item $e_{s,i} \not\equiv e_{t,i} \Mod{d}$.
  \end{enumerate}
  For the boxes falling into class 1, we group them so that all boxes $a_{k_0}, \cdots ,a_{k_{d \minu 1}}$ in a group $k$ satisfy that the number of connections between each box and each pink node is the same, that is $e_{x,k_0} = \cdots = e_{x,k_{d \minu 1}} \eqqcolon e_{x,k}$ for all pink node index $x$.
  Note that we have exactly $d$ boxes in a group, since the multiplier from the $i$-th box to the $s \wedge t$ node has weight $e_{s \wedge t,k_i} \equiv i - i \Mod{d}$, and $i$ can be set to all elements between $0$ and $d - 1$.
  Now, for all group $k$, we unfuse their connections to the pink nodes and the W on the top as follows:
  We unfuse all such groups:
  \[
    \input{qudit-nf-trace/pf2.tikz}
  \]
  then we can combine the boxes in a group the following way:
  \begin{align*}
    &\input{qudit-nf-trace/pf3-1.tikz} \\
    &\input{qudit-nf-trace/pf3-2.tikz}
  \end{align*}
  where
  \[
    \Sigma_k \coloneqq \sum_{j = 0}^{d - 1} a_{k_j}.
  \]

  After combining each box falling into class 1, we can eliminate the remaining boxes that fall into the 2-nd class.
  Let us denote the indices of class 2 boxes with $\ell_0, \cdots, \ell_p$ for some adequate $p$.
  Then, all of these boxes can be eliminated as follows:
  \begin{align*}
    &\input{qudit-nf-trace/pf4-1.tikz}\\
    &\input{qudit-nf-trace/pf4-2.tikz}
  \end{align*}
\end{proof}
\begin{remark}
  In case we take the trace a normal form with two outputs, we get the normal form of a scalar as follows:
  \[
    \input{qudit-nf-trace/scalar.tikz}
  \]
\end{remark}

%% file: qudit-tensor.tex
\nftensorproduct*
\begin{proof}
  \begin{align*}
    &\input{qudit-tensor/tensorpf1.tikz} \\
    &\input{qudit-tensor/tensorpf2.tikz} \\
    &\input{qudit-tensor/tensorpf3.tikz} \\
    &\input{qudit-tensor/tensorpf4.tikz} \\
    &\input{qudit-tensor/tensorpf5.tikz} \\
    &\input{qudit-tensor/tensorpf6.tikz} \\
    &\input{qudit-tensor/tensorpf7.tikz} \\
    &\input{qudit-tensor/tensorpf8.tikz} \\
    &\input{qudit-tensor/tensorpf9.tikz} \\
    &\input{qudit-tensor/tensorpf10.tikz} \\
    &\input{qudit-tensor/tensorpf11.tikz} \\
    &\input{qudit-tensor/tensorpf12.tikz}
  \end{align*}
  \vspace*{-1cm}
\end{proof}